\documentclass[aps,prl,amsmath,twocolumn,nobibnotes,longbibliography]{revtex4-2}
\pdfoutput=1
\usepackage[utf8]{inputenc}

\usepackage{times}
\usepackage{microtype}
\usepackage{graphicx}
\usepackage[dvipsnames]{xcolor}
\usepackage{amssymb,amsmath,bm,mathtools}
\usepackage{physics}
\usepackage{tensor}
\usepackage{adjustbox}
\usepackage{booktabs}
\usepackage{tabularx}
\usepackage{braket}
\usepackage{multirow,makecell}
\usepackage{enumitem}
\usepackage{etoolbox}
\usepackage{mathtools}
\usepackage{titlesec}
\usepackage{euscript} 
\usepackage{fontawesome}
\usepackage{dsfont}
\usepackage{pifont}
\usepackage{braket}
\usepackage{theoremref}
\usepackage{amsthm}

\newcommand{\R}{{\rho}}
\renewcommand{\S}{{\sigma}}
\newcommand{\E}{{\mathcal{E}}}
\newcommand{\GX}{{G_{\!X}}}
\renewcommand{\[}{\begin{equation}}
	\renewcommand{\]}{\end{equation}}

\newtheorem{theorem}{Theorem}
\newtheorem{lemma}[theorem]{Lemma}

\newtheorem{conjecture}[theorem]{Conjecture}

\makeatletter
\DeclareRobustCommand*{\defeq}{%
	\mathrel{\rlap{%
			\raisebox{0.3ex}{$\m@th\cdot$}}%
		\raisebox{-0.3ex}{$\m@th\cdot$}}%
	=%
}
\DeclareRobustCommand*{\eqdef}{%
	=\mathrel{\rlap{%
			\raisebox{0.3ex}{$\m@th\cdot$}}%
		\raisebox{-0.3ex}{$\m@th\cdot$}}%
}
\makeatother

\makeatletter
\g@addto@macro\bfseries{\boldmath}
\makeatother

\usepackage{etoolbox}

\makeatletter
\def\thesubsection{\thesection.\arabic{subsection}}
\def\p@subsection{}
\makeatother

\titleformat{\section}[hang]
{\normalfont\fontsize{10.5pt}{12pt}\selectfont\bfseries}
{\thesection.}{0.75em}
{\raggedright}

\titleformat{\subsection}[hang]
{\normalfont\normalsize\bfseries}
{\hspace{1em}\thesubsection.}{0.5em}
{\hspace{0pt}\raggedright}

\titleformat{\subsubsection}[hang]
{\normalfont\normalsize\bfseries}
{\thesubsubsection.}{1em}{\raggedright}

\titlespacing*{\section}{0pt}{4mm}{1mm}
\titlespacing*{\subsection}{0pt}{3mm}{0.5mm}
\titlespacing*{\subsubsection}{0pt}{2.25mm}{0.25mm}

\apptocmd{\appendix}{%
	\makeatletter
	\renewcommand{\thesubsection}{\thesection.\arabic{subsection}}%
	\makeatother
	\titleformat{\section}[display]%
	{\normalfont\normalsize\bfseries}%
	{Appendix \thesection}%
	{0.3em}%
	{\raggedright}%
	\titlespacing*{\section}{0pt}{3mm}{2.2mm}%
}{}{}

\makeatletter
\def\l@f@section{%
	\addpenalty{\@secpenalty}%
	\addvspace{0.60em plus 0.10em}%
}
\def\l@@sections#1#2#3#4{%
	\begingroup
	\everypar{}%
	\set@tocdim@pagenum\@tempboxa{#4}%
	\global\@tempdima\csname tocdim@#2\endcsname
	\@nameuse{l@f@#2}%
	\setbox\z@\vbox{%
		\hsize\linewidth
		\leftskip\csname tocleft@#2\endcsname\relax
		\dimen@\csname tocleft@#1\endcsname\relax
		\parindent-\leftskip\advance\parindent\dimen@
		\rightskip\tocleft@pagenum plus 1fil\relax
		\skip@\parfillskip\parfillskip\z@
		\let\numberline\numberline@@sections
		\ignorespaces#3\unskip\par
	}%
	\setbox\tw@\vbox to \ht\z@{%
		\vfil
		\hb@xt@\tocleft@pagenum{\hfil\unhbox\@tempboxa}%
		\vfil
	}%
	\noindent
	\rlap{\hb@xt@\linewidth{\hfil\box\tw@}}%
	\box\z@
	\par
	\expandafter\aftergroup\csname tocdim@#2%
	\expandafter\endcsname
	\expandafter\endgroup
	\the\@tempdima\relax
}
\makeatother

\makeatletter
\renewcommand{\footnoterule}{%
	\kern -2pt
	\hrule width .10\columnwidth height .4pt
	\kern 2pt
}
\makeatother

\makeatletter
\newif\ifappendixtocentries
\appendixtocentriesfalse

\newcommand{\appendixtoc}{%
	\begingroup
	\setcounter{tocdepth}{2}%
	\par\medskip
	\noindent\textbf{Contents}\par\nobreak
	\@starttoc{atoc}%
	\par\medskip
	\endgroup
}

\newcommand*\l@appsubsection{\@dottedtocline{2}{0pt}{0pt}}
\let\KM@orig@subsection\subsection
\renewcommand{\subsection}{%
	\@ifstar{\KM@orig@subsection*}{\@dblarg\KM@appsubsection}%
}

\def\KM@appsubsection[#1]#2{%
	\KM@orig@subsection[#1]{#2}%
	\ifappendixtocentries
	\addcontentsline{atoc}{appsubsection}{#1}%
	\fi
}
\makeatother

\usepackage{scalerel}
\usepackage{tikz}
\usetikzlibrary{arrows.meta,positioning,fit,calc}
\usetikzlibrary{svg.path}
\usetikzlibrary{decorations.text}
\definecolor{orcidlogocol}{HTML}{A6CE39}
\tikzset{
	orcidlogo/.pic={
		\fill[orcidlogocol] svg{M256,128c0,70.7-57.3,128-128,128C57.3,256,0,198.7,0,128C0,57.3,57.3,0,128,0C198.7,0,256,57.3,256,128z};
		\fill[white] svg{M86.3,186.2H70.9V79.1h15.4v48.4V186.2z}
		svg{M108.9,79.1h41.6c39.6,0,57,28.3,57,53.6c0,27.5-21.5,53.6-56.8,53.6h-41.8V79.1z M124.3,172.4h24.5c34.9,0,42.9-26.5,42.9-39.7c0-21.5-13.7-39.7-43.7-39.7h-23.7V172.4z}
		svg{M88.7,56.8c0,5.5-4.5,10.1-10.1,10.1c-5.6,0-10.1-4.6-10.1-10.1c0-5.6,4.5-10.1,10.1-10.1C84.2,46.7,88.7,51.3,88.7,56.8z};
	}
}
\newcommand\orcidlink[1]{\href{https://orcid.org/#1}{\mbox{\scalerel*{
				\begin{tikzpicture}[yscale=-1,transform shape]
					\pic{orcidlogo};
				\end{tikzpicture}
			}{X}}}}

\usepackage{fancyhdr}
\usepackage{soul}
\usepackage[normalem]{ulem}

\definecolor{arxivlink}{RGB}{0,140,220}
\definecolor{doilink}{RGB}{10,30,160}

\usepackage{url,hyperref}
\hypersetup{colorlinks,linkcolor={blue!55!black},citecolor={red!50!black},filecolor={blue!55!black},urlcolor={doilink},breaklinks=true}

\begin{document}
	
	\title{Connecting Quantum Tomography and Quantum Retrodiction}
	
	\author{Sebastian Murk\orcidlink{0000-0001-7296-0420}}
	\email{sebastian.murk@matfyz.cuni.cz}
	\affiliation{Faculty of Mathematics and Physics, Charles University, Ke Karlovu 3, 121 16 Praha 2, Czech Republic}
	\author{Ian Tan\orcidlink{0009-0002-0208-4407}}
	\email{y-e.tan@matfyz.cuni.cz}
	\affiliation{Faculty of Mathematics and Physics, Charles University, Ke Karlovu 3, 121 16 Praha 2, Czech Republic}
	\author{Fabian M\"uller\orcidlink{0009-0007-8159-5710}}
	\email{fabian.muller@matfyz.cuni.cz}
	\affiliation{Faculty of Mathematics and Physics, Charles University, Ke Karlovu 3, 121 16 Praha 2, Czech Republic}
	\author{Dominik \v{S}afr\'anek\orcidlink{0000-0002-6861-395X}}
	\email{dominik.safranek@matfyz.cuni.cz}
	\affiliation{Faculty of Mathematics and Physics, Charles University, Ke Karlovu 3, 121 16 Praha 2, Czech Republic}
	
	\begin{abstract}
		Quantum tomography and quantum retrodiction are traditionally viewed as separate inference tasks. We show that the two are manifestations of the same underlying principle. We prove that the Petz recovery map of a measurement channel is a gradient of the log-likelihood used in maximum-likelihood tomography, and we derive its extension to arbitrary quantum channels. Consequently, for measurement channels, repeated applications of the Petz map monotonically increase the likelihood, naturally yielding a method for quantum tomography.
	\end{abstract}
	
	\maketitle
	\thispagestyle{fancy}
	
	Quantum tomography and quantum retrodiction are very similar in spirit but differ in their mathematical formulations. Quantum tomography seeks to identify the quantum state that most likely generated a set of measurement outcomes and to design measurements that enable this reconstruction as efficiently as possible. Quantum retrodiction, by contrast, seeks to invert a generally non-invertible quantum channel by inferring the input state that would have produced a given observed output state.
	
	In quantum tomography, the task is performed by making measurements and observing either individual outcomes, probability distributions, or expectation values of observables to probe the space of density matrices. Since no single projective measurement can uniquely determine a quantum state, this typically requires measurements in multiple bases. The state reconstruction methods may differ in practice but should converge to the same solution for complete and noise-free data. The two prominent methods are linear inversion~\cite{Vogel.Risken:PhysRevA:1989,Leonhardt:Book:1997,James.etal:PhysRevA:2001,DAriano.Paris.Sacchi:AdvImagingElectronPhys:2003}, which seeks to solve for the quantum state as a solution to a set of linear equations, and maximum-likelihood estimation~\cite{Hradil:PhysRevA:1997,Banaszek.etal:PhysRevA:1999,Jezek.Fiurasek.Hradil:PhysRevA:2003,Paris.Rehacek:Book:2004,Rehacek.etal:PhysRevA:2007,BlumeKohout:PhysRevLett:2010,Goncalves:QuantumInfComput:2013,Bolduc.etal:NPJQuantumInf:2017,Aditi.Becker:PhysRevA:2025,Gaikwad.etal:QuantumSciTechnol:2025}, which aims to find the state that most-likely generated the data. Other methods include Bayesian state~\cite{Blume-Kohout:NewJPhys:2010, HuszarHoulsby:PhysRevA:2012, Kueng.Ferrie:NewJPhys:2015, Granade.Combes.Cory:NewJPhys:2016, Straupe:JetpLett:2016, Lukens.etal:NewJPhys:2020, Chapman.etal:OptExpress:2022}, shadow~\cite{Aaronson:1711.01053,Huang.Kueng.Preskill:NatPhys:2020,Stricker.etal:PRXQuantum:2022,Qin.etal:NPJQuantumInf:2025}, least squares~\cite{Guta.etal:JPhysA:2020}, gentle or weak measurements~\cite{Bennett.etal:PhysRevA:2006,Wu:SciRep:2013}, compressed-sensing~\cite{Gross.etal:PhysRevLett:2010,Flammia.etal:NewJPhys:2012}, matrix product states~\cite{Cramer.etal:NatCommun:2010,Lanyon.etal:NatPhys:2017}, self-guided~\cite{Ferrie:PhysRevLett:2014,Chapman.Ferrie.Peruzzo:PhysRevLett:2016}, and neural network tomography~\cite{Lohani.etal:MachLearnSciTechnol:2020,Torlai.etal:NatPhys:2018,Xin.etal:NPJQuantumInf:2019,Carrasquilla.etal:NatMachIntell:2019,Palmieri.etal:NPJQuantumInf:2020}.
	
	Quantum retrodiction, on the other hand, assumes knowledge of the full output quantum state, rather than only measurement outcomes or a probability distribution, together with the channel that produced it. Its task is then to infer a possible input state that could have led to the observed output. For a unitary evolution, retrodiction is trivial and amounts to evolving the observed state backward in time. For non-invertible channels, however, some information about the input is inevitably lost, so the retrodicted state cannot be uniquely determined. How to assign this missing information remains an open question, with the Petz recovery maps representing two prominent mathematically motivated approaches. Their appeal stems from the properties of exact recoverability~\cite{Petz:CommunMathPhys:1986,Petz:QuartJMathOxfordSer:1988,Petz:RevMathPhys:2003}, where the Petz map recovers the original state exactly when equality in the data-processing inequality is achieved, and approximate recoverability~\cite{Hayden.etal:CommunMathPhys:2004,Fawzi.Renner:CommunMathPhys:2015,Wilde:ProcRSocA:2015,Sutter.Fawzi.Renner:ProcRSocA:2016,Sutter.Tomamichel.Harrow:IEEETransInfTheory:2016,Sutter.Berta.Tomamichel:CommunMathPhys:2017,Junge.etal:AnnHenriPoincare:2018}, where an averaged rotated Petz map recovers a state close to the original one when the loss under data processing is small. Moreover, in the commuting case, the Petz map reduces to classical Bayesian updating motivating its interpretation as a quantum analogue of Bayes' theorem~\cite{Leifer.Spekkens:PhysRevA:2013,Tsang:PhysRevA:2022,Surace.Scandi:Quantum:2023,Parzygnat.Buscemi:Quantum:2023,Parzygnat.Fullwood:PRXQuantum:2023,Buscemi.Schindler.Safranek:NewJPhys:2023,Aw.etal:PRXQuantum:2023,Scandi.etal:RepProgPhys:2025,Bai.Buscemi.Scarani:PhysRevLett:2025,Liu.etal:PhysRevA:2025,Liu.Bai.Scarani:2510.08447,Liu.Scarani.Bai:PhysRevLett:2026}.
	
	Both quantum tomography and quantum retrodiction have critical applications in quantum technologies: quantum tomography remains the go-to method for full quantum diagnostics, enabling the characterization and verification of quantum states and processes in quantum computers, sensors, and communication devices~\cite{Chuang.Nielsen:JModOpt:1997,Poyatos.Cirac.Zoller:PhysRevLett:1997,Lvovsky.Raymer:RevModPhys:2009,Degen.Reinhard.Cappellaro:RevModPhys:2017}. It provides information about state-preparation fidelity, gate errors, noise sources, and the security of quantum communication protocols~\cite{James.etal:PhysRevA:2001,Scarani.etal:RevModPhys:2009}. Quantum retrodiction plays a central role in quantum error correction and channel reversal. The Petz map is optimal for perfectly correctable codes and remains near-optimal when the code is only approximately correctable~\cite{Barnum.Knill:JMathPhys:2002,Ng.Mandayam:PhysRevA:2010,Gilyen.etal:PhysRevLett:2022,Zheng.etal:PhysRevLett:2024,Chen.etal:2511.05941}, as supported by recent experimental demonstrations~\cite{Li.etal:2606.12020,Png.Scarani:PhysRevA:2025,Song.Kwon.Scarani:2510.26895,Singh.etal:PhysRevA:2026}. It has also found applications across quantum information theory~\cite{Beigi.Datta.Leditzky:JMathPhys:2016,Kuang.Torii.Buscemi:QuantumSciTechnol:2026}, quantum thermodynamics~\cite{Aberg:PhysRevX:2018,Kwon.Kim:PhysRevX:2019,Aw.Buscemi.Scarani:AVSQuantumSci:2021,Buscemi.Scarani:PhysRevE:2021,Buscemi.Schindler.Safranek:NewJPhys:2023,Aw.etal:PRXQuantum:2024,Nagasawa.etal:RepProgPhys:2025}, and quantum gravity~\cite{Cotler.etal:PhysRevX:2019,Chen.Penington.Salton:JHEP:2020}.
	
	In this paper, we mathematically connect these two frameworks by extending quantum retrodiction to the multi-channel setting required for tomography and extending quantum tomography beyond measurement channels to arbitrary quantum channels. We show that the gradient of the log-likelihood is represented by a quantum channel, which we call the Kubo--Mori (KM) update. It shares the Bayesian interpretation with the Petz map while being derived directly from the underlying inference problem.  For the measurement channel, when the predicted output commutes with the observed evidence, the Petz map and the KM update coincide. The Petz map itself then becomes a gradient of the standard log-likelihood. We show that the KM update increases the log-likelihood locally, prove the global monotonicity for the commuting-output case, and provide evidence for the general case. Finally, we use these results to develop and implement tomographic protocols based on iterating the Petz recovery map and the KM update for standard and general-channel quantum tomography, respectively.
	
	\paragraph*{Mathematical preliminaries.}
	We consider the setting of quantum state inference. The goal is to infer an unknown input state $\R$ from measurement data obtained through a known quantum channel $\mathcal{E}$. Since $\R$ is not directly accessible, only its image under $\mathcal{E}$ can be estimated from data. We therefore write
	\begin{align}\label{eq:X_and_Y_defs}
		X \geqslant 0,
		\qquad
		Y \defeq \mathcal{E}(\sigma) > 0,\quad \Tr[X]=\Tr[Y]=1,
	\end{align}
    where $X$ denotes the evidence and $Y$ the prediction corresponding to a candidate input state $\sigma$. The evidence may be the mean state of a classical register, an output state estimate obtained from a preceding tomographic procedure, or a state constructed from incomplete measurement data supplemented by additional physical constraints or inference principles, such as maximum or minimum entropy, entanglement, or a prescribed mean energy. $Y>0$ ensures that $\ln Y$ is well-defined.
    We assume that conditions~\eqref{eq:X_and_Y_defs} hold throughout the paper.
	
    We consider the problem of finding the state $\sigma$ that best explains the
    evidence $X$, by maximizing
    \begin{equation}
        \mathcal{L}_X(\sigma)
        \defeq
        \Tr\!\big[X\ln\mathcal{E}(\sigma)\big],
        \label{def:log-likelihood}
    \end{equation}
    which we refer to as the log-likelihood functional. Equivalently, this minimizes the Umegaki relative entropy
    $D(X\Vert\mathcal{E}(\sigma))$, which has an operational meaning through
    quantum Stein's lemma as the asymptotic rate governing optimal discrimination
    between the two states~\cite{Umegaki:KodaiMathSemRep:1962,Hiai.Petz:CommunMathPhys:1991,Ogawa.Nagaoka:IEEETransInfTheory:2000,Nielsen.Chuang:Book:2010}.
    
    For the quantum-to-classical (i.e., measurement) channel, 
	\begin{equation}
		\label{eq:qc.channel}
		\E_{\text{qc}}(\sigma) = \sum\nolimits_i \Tr[\Pi_i \sigma] \ketbra{i}{i},
	\end{equation}
	which maps a quantum state to measurement outcomes encoded in a classical register, the log-likelihood functional \eqref{def:log-likelihood} reduces to the standard form,
	\begin{equation}
		\mathcal{L}_{\text{qc}}(\sigma) = \sum\nolimits_i \hat{p}_i \ln \Tr[\Pi_i \sigma],
		\label{eq:log-likelihood.qc.channel}
	\end{equation}
	used in maximum-likelihood quantum tomography, where $\hat{p}_i=N_i/N$ are relative measured frequencies, $\{\Pi_i\}$ is a set of POVM elements, and $\{\ket{i}\}$ is an orthonormal basis of the classical register.

    \emph{KM update as a gradient of the log-likelihood.}
    A gradient of the log-likelihood~\eqref{def:log-likelihood}, with an appropriate choice of the input metric, defines a quantum channel,
	\begin{align}
		\mathcal{R}^{\text{KM}}_{\sigma,\mathcal{E}}(X)
		\defeq
		\sqrt{\sigma}\,
		\mathcal{E}^{\dagger}\!\big[\Omega_Y^{-1}\!(X)\big]
		\sqrt{\sigma},
		\label{def:KM.update}
	\end{align}
    which maps the evidence state with a candidate input state into an updated input state $\sigma_+$.	We call this map the KM update. $\Omega_Y^{-1}$ denotes the inverse KM superoperator. See the End Matter for details.
	
	\paragraph*{Relation between the KM update and the Petz recovery map.}
	The Petz recovery map is defined as~\cite{Petz:QuartJMathOxfordSer:1988},
	\begin{align}
		\mathcal{R}_{\sigma,\mathcal{E}}(X)
		\defeq
		\sqrt{\sigma}\,
		\mathcal{E}^{\dagger}\!\big[
		Y^{-1/2}XY^{-1/2}
		\big]
		\sqrt{\sigma},
		\label{def:Petz.update}
	\end{align}
	with $X$ and $Y$ as in Eq.~\eqref{eq:X_and_Y_defs}.
	We have the following relation.
	\begin{theorem} \label{thm:KM.update}
		The KM update can be expressed as,
		\begin{equation}
			\mathcal{R}^{\text{KM}}_{\sigma,\mathcal{E}}(X) = 
			\sqrt{\sigma}\, \mathcal{E}^{\dagger}\!\Big[ k(K_Y) \big[Y^{-1/2} X Y^{-1/2}\big]\Big] \sqrt{\sigma},
		\end{equation}
		where $K_Y \defeq [\ln Y,\cdot]$ encodes the noncommutativity between the evidence $X$ and the prediction $Y$, 
		and
		\begin{equation}
			k(\omega) = \frac{\omega}{2\sinh(\omega/2)} = 1 - \frac{\omega^2}{24} + \mathcal{O}(\omega^4).
			\label{eq:KM.scalar.filter}
		\end{equation}
	\end{theorem}
	Thus, the Petz map is the zeroth-order approximation to the KM update. Since $k$ is an even function, the leading noncommutative correction is quadratic in $K_Y$. This is why, for weak noncommutativity, the Petz map and the KM update nearly coincide. Although all these maps reduce to the ordinary Petz map in their respective commuting limits, the KM update is distinct both from the map proposed in Ref.~\cite{Bai.Buscemi.Scarani:PhysRevLett:2025}, which is derived from a fidelity criterion for matching forward and reverse processes, and from the weighted rotated Petz map~\cite{Sutter.Tomamichel.Harrow:IEEETransInfTheory:2016}, despite using the same weights $\beta_0$ (see Eq.~\eqref{def:KM.superoperator}). Additional properties of the KM update, together with all proofs, are provided in the Supplementary Material (SM)~\cite{SM}. The underlying geometric intuition is illustrated in Fig.~\ref{fig:KM.update.schematic}.

    \begin{figure}[t!]
    \centering
    \includegraphics[width=\columnwidth]{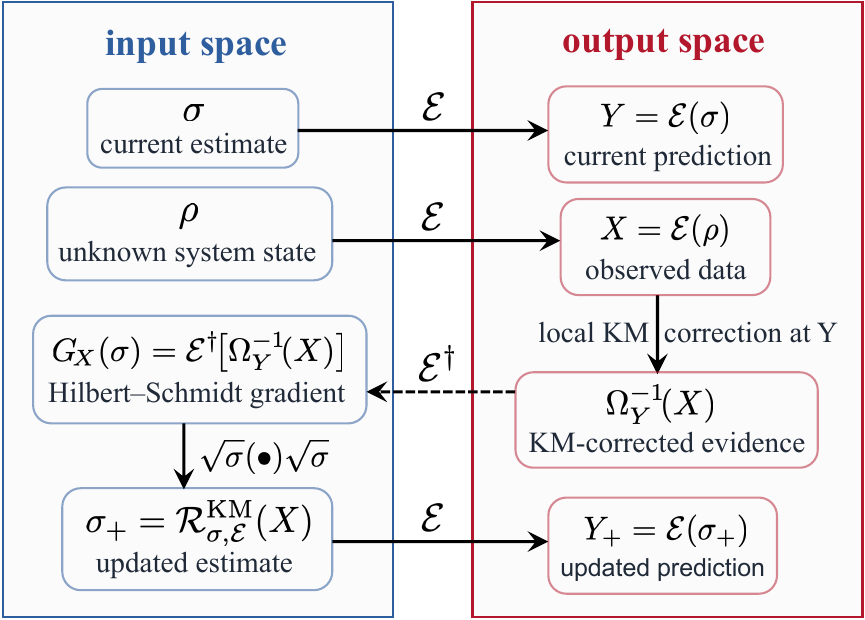}
    \begingroup
    \setlength{\abovecaptionskip}{-3mm}
    \setlength{\belowcaptionskip}{-4mm}
    \caption{Schematic representation of the KM update. The current estimate $\sigma$ predicts $Y=\mathcal{E}(\sigma)$, while the data define an output-space evidence operator $X$; ideally, $X=\mathcal{E}(\rho)$. The inverse KM action at $Y$ gives the corrected evidence $\Omega_Y^{-1}\!(X)$, which is pulled back through $\mathcal{E}^{\dagger}$ to the Hilbert--Schmidt gradient $\GX(\sigma)$. Conjugation by $\sqrt{\sigma}$ produces the positive state update $\sigma_+=\mathcal{R}^{\text{KM}}_{\sigma,\mathcal{E}}(X)$, with updated prediction $Y_+=\mathcal{E}(\sigma_+)$.}
    \label{fig:KM.update.schematic}
    \endgroup
	\end{figure}
	
	\paragraph*{Monotonicity.}
	We define a line-search update using a Riemannian gradient $\Delta$,
    \[
    \sigma_\eta = \sigma + \eta\Delta= (1-\eta)\,\sigma + \eta\,\sigma_+,
    \]
    where $\sigma_+=\mathcal{R}^{\text{KM}}_{\sigma,\mathcal{E}}(X)$ corresponds to the full step, $\eta=1$. This update in the density matrix is in the direction of the maximum log-likelihood increase (see the End Matter), which guarantees that the log-likelihood will increase locally. We derive an exact quantitative bound showing this.
	
	\begin{theorem}\label{minimum_increase}
		The minimal and maximal guaranteed increase in the log-likelihood is given by,
		\begin{equation}
			-\ln\!\big(1\!-\! \eta \textsl{g}_{\sigma}\!(\Delta_{\sigma_\eta}^\sigma,\Delta)\big) \leqslant \mathcal{L}_X(\sigma_\eta) \!-\! \mathcal{L}_X(\sigma) \leqslant \eta \textsl{g}_{\sigma}(\Delta,\Delta),
		\end{equation}
		where $\!\Delta_{\sigma_\eta}^\sigma \!\!\!\defeq\!\!\! G_{\!X}^\sigma(\sigma_\eta) \!\!-\!\! \sigma$, $\!\Delta \!\!\!=\!\!\! \Delta_{\sigma}^\sigma$, and $\! G_{\!X}^\sigma(\sigma_\eta) \!\!=\!\! \sqrt{\sigma} \mathcal{E}^\dagger[\Omega_{\mathcal{E}(\sigma_\eta)}^{-1}\!(X)] \sqrt{\sigma}$. For non-invertible $\sigma$, the expressions are evaluated by replacing $\sigma^{1/2}\sigma^{-1/2} \!\to\! \mathds{1}$, which leads to $-\ln \Tr\!\big[X\Omega_{\mathcal{E}(\sigma_\eta)}^{-1}\!(Y)\big]$ and $\eta\! \Tr\!\big[\mathcal{E}^\dagger[\Omega_{\mathcal{E}(\sigma)}^{-1}(X)]\Delta]$ on the left and right-hand sides, respectively.
	\end{theorem}
	This theorem shows that, unless the Riemannian gradient $\Delta=0$ (i.e., $\sigma$ is a fixed point of the KM update; a stationary point of the log-likelihood for an invertible $\sigma$; see the End Matter), there always exists a sufficiently small $\eta>0$ such that the argument of the logarithm is less than one. Consequently, the log-likelihood increases strictly. Inspired by this, we formulate the following conjecture: the KM update increases the log-likelihood globally.
	\begin{conjecture}\label{con:globally.monotone.KM.update}
		The KM update can only increase the log-likelihood,
		\begin{equation}
			\mathcal{L}_X(\sigma_+) \geqslant \mathcal{L}_X(\sigma),
		\end{equation}
		where $\sigma_+ = \mathcal{R}^{\text{KM}}_{\sigma,\mathcal{E}}(X)$.
		Equality holds iff $\sigma$ is the fixed point of the KM update, $\sigma_+=\sigma$. 
	\end{conjecture}
	Numerical tests on $10^8$ random instances in dimensions $2\!-\!5$ and Kraus ranks $2\!-\!5$, supplemented by 40 independent simulated-annealing runs, revealed no numerically significant violation of $\Delta_{\mathrm{KM}}\geqslant 0$, with the smallest value reaching zero within machine precision; see the SM for details~\cite{SM}. We prove a weaker version of the conjecture for quantum-to-classical channels. For such channels, $X$ and $Y$ commute, and the KM update reduces to the Petz map.

    \begin{figure}[t]
		\centering
		\includegraphics[width=\columnwidth]{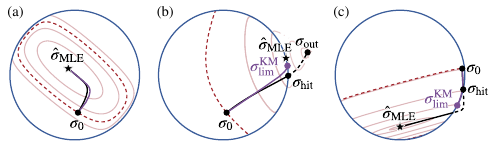}
		\begingroup
		\setlength{\abovecaptionskip}{-1.5mm}
		\caption{
        Hilbert–Schmidt gradient flow (black) vs.\ full-step KM iteration (purple). (a) Starting sufficiently close to the maximum leads to convergence. (b) When the unconstrained maximum lies outside the physical state space, the KM trajectory may approach a nearby boundary point. (c) Even when the maximum lies inside the physical state space, KM iterations starting outside its basin of attraction can converge to a nonoptimal boundary state. See Theorems~\ref{thm:Petzqc} and~\ref{thm:Petz.converg}.}
		\label{fig:convergence}
		\endgroup
	\end{figure}
	
	\begin{theorem}[Petz monotonicity for quantum-to-classical likelihood]\label{thm:Petzqc}
		Consider the quantum-to-classical channel \eqref{eq:qc.channel}. For any state $\S$, 
		the Petz update is given by,
		\begin{align}\label{eq:Petz_update}
			\sigma_+
			\defeq
			\mathcal{R}_{\sigma,\mathcal{E}_{\text{qc}}}(X)
			=
			\sum\nolimits_i\,
			\frac{\hat{p}_i}{\Tr[\Pi_i\sigma]}\,
			\sqrt{\sigma}\;\Pi_i\,\sqrt{\sigma}.
		\end{align}
		Then the standard tomographic log-likelihood satisfies
		\begin{align}
			\mathcal{L}_{\text{qc}}(\sigma_+)
			\geqslant
			\mathcal{L}_{\text{qc}}(\sigma).
		\end{align}
		Equality holds iff $\sigma$ is a fixed point of the Petz map, $\sigma_+=\sigma$.
	\end{theorem}
	Along with the Petz recovery map~\cite{Bai.Buscemi.Scarani:PhysRevLett:2025}, the KM update is a possible generalization of Jeffrey's rule to quantum physics. The above theorem and the previous conjecture imply that the application of this version of Jeffrey's rule always increases the log-likelihood. This generalizes the classical result of Jacobs~\cite{Jacobs:ElectronProcTheorComputSci.2021}, see SM~\cite{SM}.
	
	Note that replacing the KM update with the Petz recovery map in Conjecture \ref{con:globally.monotone.KM.update} does not satisfy the desired inequality: since the Petz map is not the exact gradient of the log-likelihood \eqref{def:log-likelihood}, applying it at a maximizer for noncommuting outputs with overcomplete basis and imperfect data can decrease it. A concrete example exhibiting a large violation is provided in the SM~\cite{SM}.
	
	\paragraph*{KM update and Petz recovery as quantum tomography methods.}
	The following observation, which is a rephrasing of the discussion below Eq.~\eqref{eq:stationary.condition} combined with concavity of the log-likelihood (see SM~\cite{SM}), provides the basis for the KM update as a tomography method.
	\begin{theorem}\label{thm:KM.fixed.point.MLE}
		Every full-rank fixed point $\sigma$ of the KM update is a maximizer of the log-likelihood (MLE).
	\end{theorem}
	
	Theorem~\ref{minimum_increase} guarantees that a line-search implementation can choose a step that increases the log-likelihood whenever the current state is not stationary. 	By Theorem~\ref{thm:KM.fixed.point.MLE}, any full-rank fixed point reached by the iteration is a maximizer.
	
	For the full step $\eta = 1$ in the Petz case, we prove monotonicity and subsequential fixed-point convergence; under uniqueness of the relevant fixed point, the whole sequence converges. We conjecture that the analogous full-step statement extends to the KM update.
	
	\begin{theorem}[Convergence of the Petz iteration]\label{thm:Petz.converg}
		Suppose that $\{\Pi_i\}$ is a tomographically complete  measurement with all observed probabilities $\hat{p}_i$ positive. We denote the unique maximizer of the log-likelihood $\mathcal{L}_{\text{qc}}$ as $\hat{\sigma}_{\text{MLE}}$, and the iterated Petz update as $\sigma_k=(\sigma_{k-1})_+$, for some initial prior $\sigma_0$. Then $\sigma_k \to \hat{\sigma}_{\text{MLE}}$ or $\det(\sigma_k)\to 0$. If the upper level set $\{\sigma:\mathcal{L}(\sigma)\geqslant \mathcal{L}(\sigma_0)\}$ contains only invertible states, then $\sigma_k\to\hat{\sigma}_{\text{MLE}}.$
	\end{theorem}
	
    While the above results apply to a single measurement, we extend them to any number of measurements by mapping them onto a single effective POVM. This yields a general method for quantum tomography as follows. Consider a set $\{\mathcal{E}_j\}$ of channels with corresponding estimated outputs $X_j$. Define the weighted log-likelihood as
	\begin{equation}
		\sum\nolimits_j w_j \mathcal{L}_{X_j}(\sigma),
		\quad
		\mathcal{L}_{X_j}(\sigma) = \Tr\bigl[X_j \ln \mathcal{E}_j(\sigma)\bigr],
	\end{equation}
	which generalizes the standard maximum-likelihood formulation. The weights $\{w_j\}$ represent the relative importance of each dataset and, in analogy with the usual log-likelihood, are typically chosen as empirical frequencies $w_j = N_j / N$ of channel executions yielding data $X_j$. 
	
	Due to the linearity of the Fréchet derivative, we can define the weighted KM update as
	\begin{equation}
		\mathcal{R}^{\text{KM}}_{\sigma,\{\mathcal{E}_j\}}(\{X_j\}) 
		\defeq 
		\sum\nolimits_j w_j \mathcal{R}^{\text{KM}}_{\sigma,\mathcal{E}_j}(X_j),
		\label{def.weighted.KM.update}
	\end{equation}
	which corresponds to the gradient of the weighted log-likelihood in the inverse-square-root metric \eqref{def:inverse.square.root.metric}. Equivalently, this is the ordinary KM update for the block channel $\mathcal{E}_\oplus = \bigoplus_j w_j \mathcal{E}_j$ with data $X_\oplus=\bigoplus_j w_j X_j$, up to an irrelevant additive constant in the likelihood, see SM~\cite{SM}. 

	\begin{figure*}[htbp!]
		\centering
		\includegraphics[width=.48\linewidth]{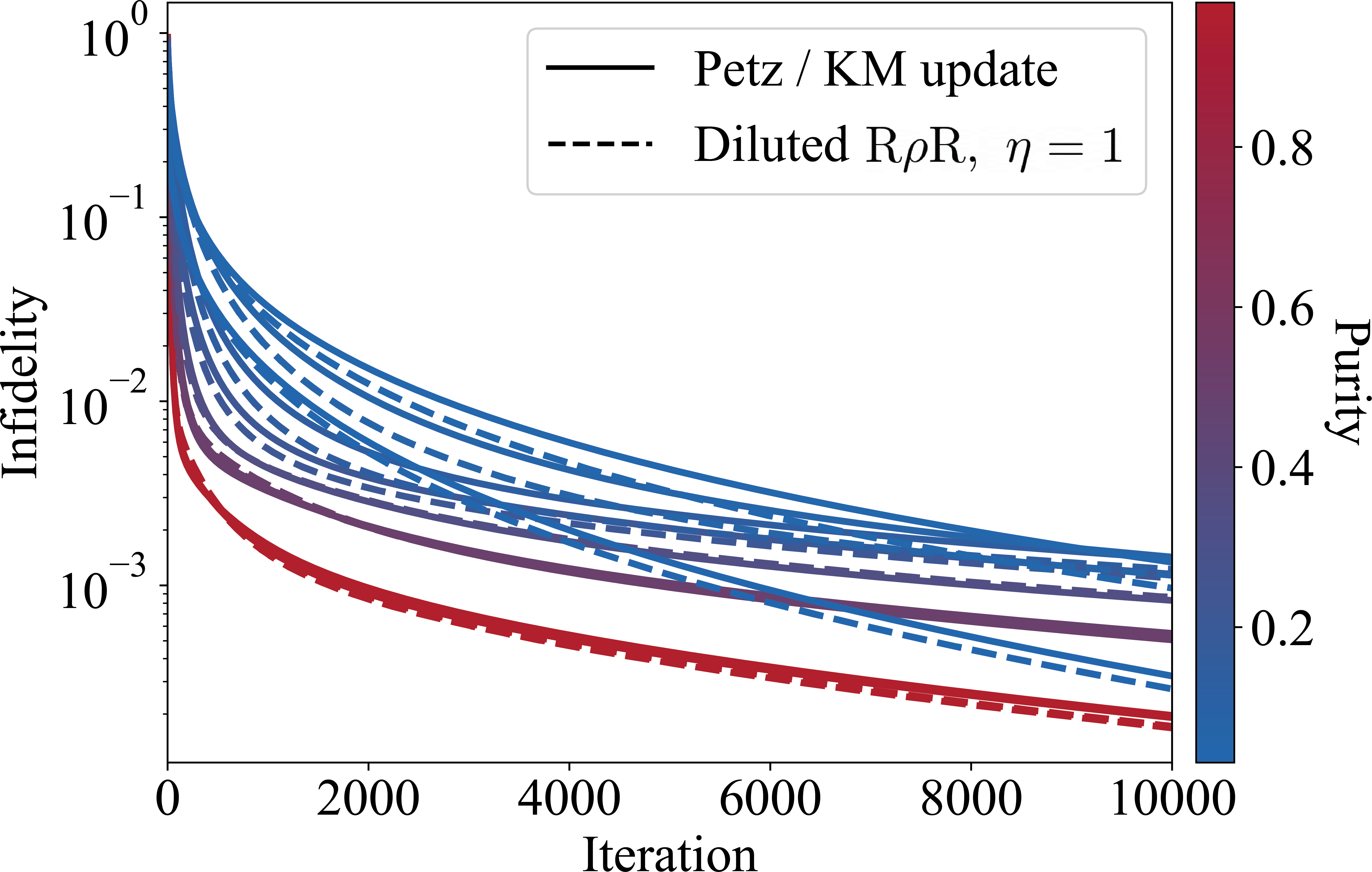}
		\hspace{10pt}
		\includegraphics[width=.48\linewidth]{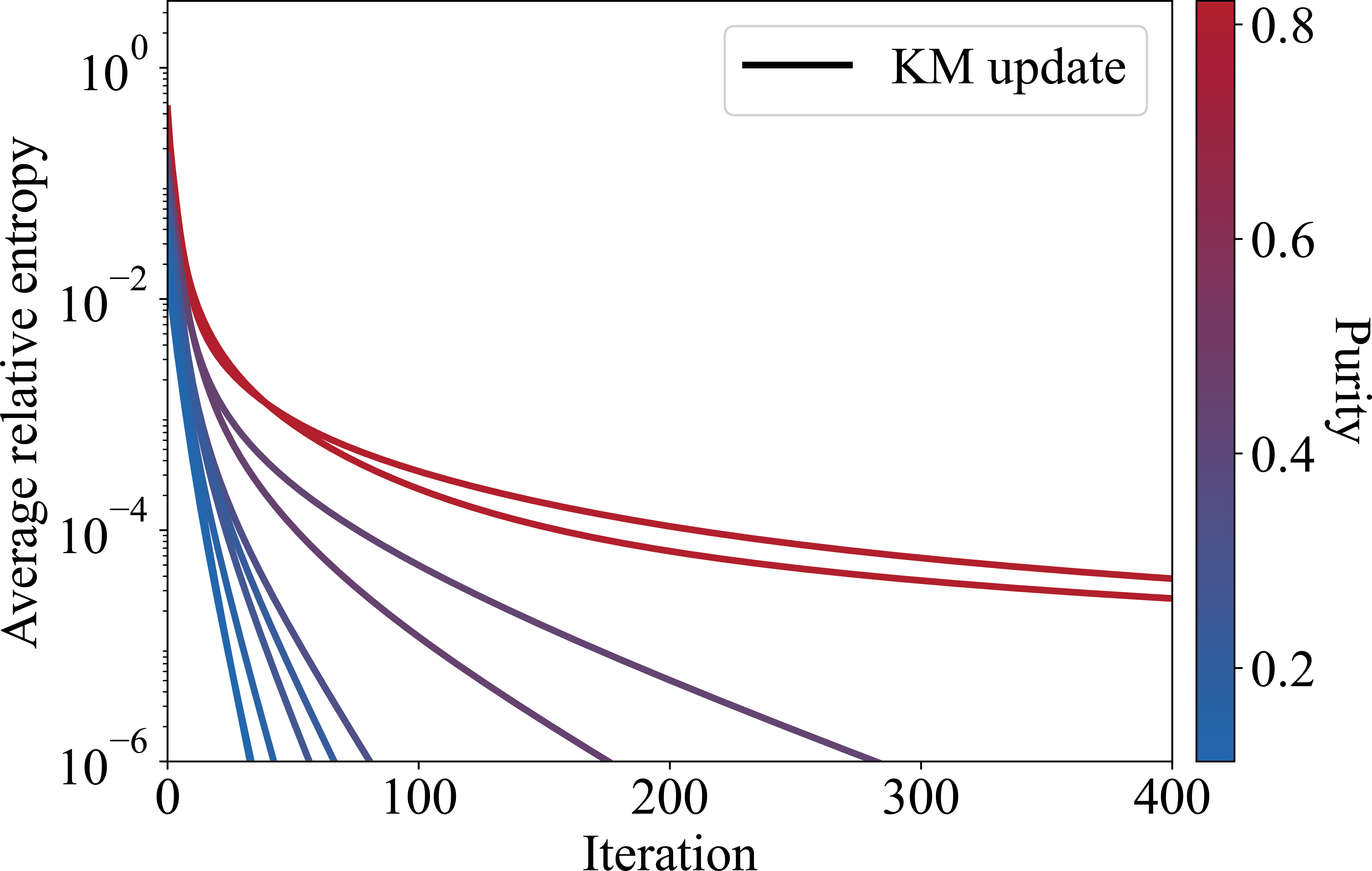}
		\begingroup
		\setlength{\abovecaptionskip}{2mm}
		\setlength{\belowcaptionskip}{-3mm}
    	\caption{Left: Maximum-likelihood quantum tomography based on iterative application of the weighted Petz map over Pauli bases with full $3^6$ measurement settings, with the prior initialized as the maximally mixed state. Ten Wishart-random~\cite{Zyczkowski.Sommers:JPhysA:2001} six-qubit density matrices of various ranks were sampled and regularized to make them full-rank, $\R=(1-\epsilon)\R+\epsilon\, I/d$, $\epsilon = 10^{-2}$, to put the initial prior into the basin of attraction (see Fig.~\ref{fig:convergence}). Diluted $\mathrm{R}\rho\mathrm{R}$ shows comparable performance while being slightly faster, with a KM-to-diluted $\mathrm{R}\rho \mathrm{R}$ wall-clock time ratio of $1.23\pm0.13$. Right: Multi-view KM reconstruction of a global four-qubit state consistent with six supplied two-qubit reduced density matrices. Ten test states were sampled as in the previous example and regularized with $\epsilon=0.1$. Both examples operate with exact probabilities.}
		\label{fig:petz50}
		\endgroup
	\end{figure*}	
    
    The channel set may be incomplete, overcomplete, or informationally complete. Quantum tomography is performed by iterating the KM update from an initial prior $\S_0$. The prior and weights influence the trajectory and its limit, which matters especially for noisy or incomplete data. Figure~\ref{fig:convergence} illustrates both convergence to an MLE and approach to a nonoptimal boundary state, which can occur even for compatible data. In the overcomplete case, the weights also encode the relative importance of different channels, prioritizing those with more reliable data.
	
	We illustrate the method in Fig.~\ref{fig:petz50} on the standard tomography setup using Pauli measurements on six qubits and on the following noncommutative KM example with incomplete information, inspired by local quantum tomography~\cite{Torlai.etal:NatPhys:2018}: Consider an experimenter who performed full tomography on every pair of qubits of a four-qubit system, determining the six reduced density matrices, $
    X_{AB},X_{AC},X_{AD},X_{BC},X_{BD},X_{CD}$,
    This corresponds to the six quantum channels,
    \begin{equation}
    \mathcal{E}_{AB}(\rho)=\Tr_{CD}[\rho],\ \dots
    \end{equation}
    We seek which global state $\sigma$ best reproduces the available local data. Since the six reduced density matrices do not uniquely determine the global state, we infer a single $\sigma$ by iterating the KM update from the maximally mixed prior. We quantify convergence by the average quantum relative entropy
    $\frac{1}{6}\sum_{ij}
    D\!\left(
    X_{ij}\middle\Vert\mathcal{E}_{ij}(\sigma)
    \right)$,
    which vanishes when all six reconstructed marginals reproduce the measured reduced states. Because the two-qubit reductions overlap, $X_{ij}$ and $\mathcal{E}_{ij}(\sigma)$ generally do not commute during the iteration, providing a genuinely noncommutative application of the KM update. An additional thermodynamic example combining measured and inferred data is given in the SM.
    
	\paragraph*{Discussion and conclusions.}
	In this paper, we studied the exact gradient of the log-likelihood functional and found that it leads to an update rule that we call the KM update. Together with the Petz recovery map and the rotated Petz recovery map, the KM update may be viewed as a candidate quantum analogue of Bayes' and Jeffrey's rule. All three maps reduce to the classical Bayes update in the commuting case. Unlike the Petz map, which arises naturally in the characterization of equality in the data-processing inequality~\cite{Petz:CommunMathPhys:1986,Petz:RevMathPhys:2003}, or the rotated Petz map, which strengthens this connection through fidelity-based recovery bounds~\cite{Sutter.Tomamichel.Harrow:IEEETransInfTheory:2016,Sutter.Berta.Tomamichel:CommunMathPhys:2017}, the KM update arises from the maximum-likelihood inference problem, namely the task of finding the reference state whose predicted output best explains the observed data.
	
	Specifically, if closeness is quantified by the Umegaki quantum relative entropy (a choice supported by its uniqueness under four natural axioms~\cite{Wilming.Gallego.Eisert:Entropy:2017}), then the KM update is its gradient with respect to the reference state and therefore follows the direction of steepest descent. Conjecture~\ref{con:globally.monotone.KM.update} and Theorem~\ref{thm:Petzqc} further suggest that it monotonically decreases the relative entropy, extending Jacobs' monotonicity result for Jeffrey's rule and the Kullback--Leibler divergence~\cite{Jacobs:ElectronProcTheorComputSci.2021}. In contrast, the Petz recovery map violates the analogous monotonicity property.
	
	Since the KM update points in the direction of the likelihoods' maximal increase, it satisfies local monotonicity, which guarantees that the limit exists. Theorems~\ref{minimum_increase} and~\ref{thm:KM.fixed.point.MLE} establish sufficient conditions for this limit to be the MLE. In the commuting-output setting, the KM iteration reduces to the Petz iteration, for which we establish the stronger properties of full-step monotonicity and convergence to the maximum-likelihood estimator in Theorems~\ref{thm:Petzqc} and~\ref{thm:Petz.converg}.
	
	This allows us to design and implement a tomographic protocol that is a first-order gradient-ascent method, applicable both to standard quantum tomography via the quantum-to-classical channel and to a multi-view reconstruction using arbitrary quantum channels. This is realized by formulating the weighted KM update \eqref{def.weighted.KM.update}, which enables us to treat incompatible measurements as well as undercomplete, informationally complete, overcomplete, and imperfect data within a unified framework.

    For quantum-to-classical channels, the Petz update becomes $\sigma_+=\sqrt{\sigma}R\sqrt{\sigma}$, where $R=\sum_i\frac{\hat p_i}{\Tr[\Pi_i\sigma]}\Pi_i$, and can be compared with the standard and diluted $R\rho R$ iterations, $\sigma_+^{R\rho R}=\mathcal N(R\sigma R)$ and $\sigma_+^{R\rho R(\eta)}=\mathcal N(R_\eta\sigma R_\eta)$, with $R_\eta=(R+\eta I)/(1+\eta)$. For commuting $R$ and $\sigma$, the methods reduce to $\sigma_+=\sigma R$ and $\sigma_+^{R\rho R}=\sigma R^2/\Tr[\sigma R^2]$. Monotonicity guarantees the convergence of the Petz update, to a MLE if the limit point is full-rank, whereas R$\rho$R requires adaptive-$\eta$ schemes for convergence guaranties~\cite{Rehacek.etal:PhysRevA:2007,Goncalves:QuantumInfComput:2013}. The Petz update may nevertheless converge to a non-MLE boundary point. Unlike R$\rho$R, the Petz update is a CPTP map on the evidence state. The Petz update has a slightly higher per-iteration computational cost due to the additional evaluation of $\sqrt{\sigma}$.
	
	Another interesting comparison can be made with Bayesian quantum tomography, which applies the classical Bayes theorem to infer a posterior distribution over quantum states. A common point estimate is the Bayesian mean estimator,
	$
	\hat{\sigma}_{\text{BME}}
	= \int \sigma\, p(\sigma | \text{data})\, d\sigma,
	$
	obtained by averaging over the posterior distribution~\cite{Blume-Kohout:NewJPhys:2010, HuszarHoulsby:PhysRevA:2012, Kueng.Ferrie:NewJPhys:2015, Granade.Combes.Cory:NewJPhys:2016, Straupe:JetpLett:2016, Lukens.etal:NewJPhys:2020, Chapman.etal:OptExpress:2022}. In contrast to the present approach, this integral is generally not available in closed form and is typically approximated using methods such as sequential Monte Carlo or Metropolis--Hastings sampling~\cite{Struchalin.etal:PhysRevA:2016}. 
	
	Here, we find the KM update to generalize Bayes' and Jeffrey's rules, while being formulated entirely in terms of quantum states rather than probability distributions. Starting from the search for a fully quantum analogue of Bayesian inference, we therefore arrived at a somewhat surprising conclusion: the resulting update naturally leads to maximum-likelihood estimation. This raises the intriguing possibility that maximum-likelihood estimation is, in fact, the natural form of Bayesian inference in quantum theory.
	
	The remaining open problems are to establish the global monotonicity of the undamped KM update, derive general convergence guaranties, characterize its limiting points even in undercomplete settings, and develop faster and higher-order tomographic protocols inspired by the present approach, which will be reported in a separate work~\cite{Muller.Murk.Safranek:inprep}. 

    Finally, the results presented here reveal a geometric structure underlying maximum-likelihood inference, arising from the interplay between Kubo--Mori geometry at the output and inverse-square-root geometry at the input, connecting quantum tomography, recovery maps, and quantum Bayesian retrodiction within a unified framework.
	
	\section*{Acknowledgments}
	This research was supported by the Czech Science Foundation through the Junior Star grant 25-17250M (SM, FM, and D\v{S}) and by Charles University in Prague through PRIMUS/25/SCI/027 (IT and D\v{S}). We thank Joe Schindler and Valerio Scarani for the initial discussions that sparked this project, and especially Francesco Buscemi for early insights into the possible convergence of the Petz recovery map iteration.
	
	\section*{Data availability}
	The source code used to generate Fig.~\ref{fig:convergence} is available at
    \href{https://github.com/sebastianmurk/2606.23777}{https://github.com/sebastianmurk/2606.23777}.
    The source code used to generate Fig.~\ref{fig:petz50}, the numerical evidence supporting Conjecture~\ref{con:globally.monotone.KM.update}, and the search for the Petz monotonicity counterexample reported in the SM~\cite{SM} are available at \href{https://github.com/dominik-safranek/km-update-tomography}{https://github.com/dominik-safranek/km-update-tomography}.

    \appendix

    \section{End Matter}

    \paragraph*{Log-likelihood gradient.}
	The KM superoperator and its inverse are defined as~\cite{Kubo.Tomita:JPhysSocJpn:1954,Mori:ProgTheorPhys:1965,Petz.Toth:LettMathPhys:1993,Petz:JMathPhys:1994,Petz:LinearAlgebraAppl:1996}	\begin{equation}\label{def:KM.superoperator}
		\begin{split}
			\Omega_Y(X) &= \int_0^1 Y^t X\, Y^{1-t}\,dt,\\		
			\Omega_Y^{-1}\!(X) &= \int_{\mathbb{R}} \beta_0(t)\, Y^{-(1+it)/2}X\,Y^{-(1-it)/2}\,dt,
		\end{split}		
	\end{equation}
	where $\beta_0(t) = \pi/[2(\cosh(\pi t)+1)]$.
	To understand how small changes in $\sigma$ affect the log-likelihood functional \eqref{def:log-likelihood}, we consider its Fréchet derivative,
	\begin{equation}
		D \mathcal{L}_X(\sigma) [\Delta] = \Tr\!\big[ X\, \Omega^{-1}_Y\!\big(\mathcal{E}(\Delta)\big) \big].
		\label{eq:Fréchet.derivative}
	\end{equation}
	It can be interpreted as the directional derivative of $\mathcal{L}_X$ at $\sigma$ in the direction $\Delta$. The Fréchet derivative of the matrix logarithm at $Y$ is the inverse KM action $\Omega_Y^{-1}$~\footnote{Equivalently, if $f(A) \defeq \ln A$, then the Fréchet derivative of $f$ at $Y$ is the linear map $D f(Y) \eqdef \Omega^{-1}_Y$, i.e., $Df(Y)[X]=\Omega^{-1}_Y\!(X)$ \cite{Higham:Book:2008}.}.
	
	The gradient is the Riesz representation of the Fréchet derivative with respect to a chosen inner product~\cite{Conway:Book:2007}. We consider two inner products: the Hilbert--Schmidt inner product and, for a full-rank state $\sigma$, the inverse-square-root metric
	\begin{equation}
		\textsl{g}_\sigma(A,B) \defeq \Tr\!\big[\sigma^{-1/2} A^\dagger \sigma^{-1/2} B\big],
		\label{def:inverse.square.root.metric}
	\end{equation}
	which is a monotone metric corresponding to $f=\sqrt{\cdot}$ \cite{Petz:LinearAlgebraAppl:1996} and a metric dual to the KMS metric~\cite{Junge:JStatPhys:2023,Vernooij.Wirth:CommunMathPhys:2023}. In the commuting case, the latter reduces to the classical Fisher metric $\textsl{g}_p(a,b) \defeq \sum_i \frac{a_i b_i}{p_i}$. We choose this metric because the corresponding gradient defines a CPTP map related to the Petz map. Since $\Omega^{-1}_Y$ is self-adjoint with respect to the Hilbert--Schmidt inner product, and $\mathcal{E}^\dagger$ is the Hilbert--Schmidt adjoint of $\mathcal{E}$, we have
	\begin{equation}
		D \mathcal{L}_X(\sigma)[\Delta] = \Tr\left[\GX(\sigma)\, \Delta\right] = \textsl{g}_\sigma\!\,\big(G_{\!X}^\sigma(\sigma),\Delta\big),
		\label{eq:Fréchet.derivative.gradient}
	\end{equation}
	where the two gradients with respect to $\sigma$ for a fixed output operator $X$ are defined as
	\begin{equation}
		\GX(\sigma) \defeq \mathcal{E}^{\dagger}\!\big[\Omega^{-1}_{\mathcal{E}(\sigma)}\!(X)\big],
		\quad 
		G_{\!X}^\sigma(\sigma) \defeq \sqrt{\sigma}\, \GX(\sigma) \sqrt{\sigma}. 
		\label{def:gradient}
	\end{equation}
	
	A full-rank stationary point of $\mathcal{L}_X$ on the trace-one manifold, such as a maximizer of the log-likelihood, satisfies
	\begin{equation}
		D\mathcal{L}_X(\S)[\Delta]=0,\quad\text{ for all }\;\; \Delta=\Delta^\dagger,\; \Tr[\Delta]=0.
	\end{equation}
	Using \eqref{eq:Fréchet.derivative}, this is equivalent to
	\begin{equation}
		\label{eq:stationary.condition}
		\GX(\S)=\lambda \mathds{1},
	\end{equation}
	for some Lagrange multiplier $\lambda\in\mathbb{R}$. Thus \eqref{eq:stationary.condition} is the Euler--Lagrange equation for the trace-one constrained optimization problem. Using that both $X$ and $\sigma$ are density matrices implies $\lambda=1$: $\Tr[\GX(\S)\S]=\Tr[\Omega_Y^{-1}\!(X)Y]=\Tr[X]$. If $\sigma$ is full rank, then Eq.~\eqref{eq:stationary.condition} is equivalent to $G_{\!X}^\sigma(\sigma)=\sigma$. This can be proven by multiplying the equation by $\sqrt{\sigma}$ and $\sigma^{-1/2}$ from both sides to get the right and left implications, respectively. 
	
	\paragraph*{KM update as the direction of maximal log-likelihood increase.}
	Let us define the line-search update function
	\begin{equation}
		\sigma_\eta = \sigma + \eta\Delta,
		\qquad 
		f(\eta) \defeq \mathcal{L}_X(\sigma_\eta).
	\end{equation}
	From the definition of the Fréchet derivative, we have
	\begin{equation}
		f'(0) = D\mathcal{L}_X(\sigma)[\Delta] = \textsl{g}_\sigma\big(G_{\!X}^\sigma(\sigma) - \sigma,\Delta\big),
		\label{eq:directional.derivative.gradient}
	\end{equation}
	where $\Tr[\Delta]=0$ allows us to make the first argument tangent by adding a $-\sigma$ term. 
	In the inverse-square-root metric \eqref{def:inverse.square.root.metric}, the steepest-ascent direction on the trace-one manifold is given by the Riemannian gradient,
	\begin{equation}
		\Delta = G_{\!X}^\sigma(\sigma)-\sigma.
		\label{eq:delta.gradient}
	\end{equation}
	Substituting this choice into \eqref{eq:directional.derivative.gradient} gives
	\begin{equation}
		f'(0) = \textsl{g}_\sigma(\Delta,\Delta) = \|\Delta\|_{\textsl{g}_\sigma}^2 \geqslant 0.
	\end{equation}
	For full-rank $\sigma$, equality holds iff $\Delta=0$, i.e., iff the stationary condition \eqref{eq:stationary.condition} is satisfied with $\lambda=1$. 
	
	Using \eqref{eq:delta.gradient}, the line search is given by
	\begin{equation}
		\sigma_\eta = (1-\eta)\,\sigma + \eta\,\sigma_+,
		\label{eq:update_dm}
	\end{equation}
	where $\sigma_+=G_{\!X}^\sigma(\sigma)$ denotes the updated density matrix corresponding to $\eta=1$, in which case the update becomes purely multiplicative.
	
	The prescription $X\to\sigma_+$ defines a quantum channel,
	\begin{align}
		\mathcal{R}^{\text{KM}}_{\sigma,\mathcal{E}}(X)
		\defeq
		G_{\!X}^\sigma(\sigma)
		=
		\sqrt{\sigma}\,
		\mathcal{E}^{\dagger}\!\big[\Omega_Y^{-1}\!(X)\big]
		\sqrt{\sigma}.
		\label{def:KM.update.endmatter}
	\end{align}
	We call this map the KM update.
	
	\bibliographystyle{bibstyle}
	\nocite{Higham:Book:2008}
	\bibliography{KM_references}

@misc{SM,
	note = {See Supplemental Material}
}

@article{Aaronson:1711.01053,
	author = {S. Aaronson},
	title = {Shadow Tomography of Quantum States},
	journal = {},
	volume = {},
	pages = {},
	year = {2017},
	doi = {10.48550/arXiv.1711.01053},
	archivePrefix = {arXiv},
	eprint = {1711.01053},
	primaryClass = {quant-ph}
}

@article{Liu.Bai.Scarani:2510.08447,
	author = {M. Liu and G. Bai and V. Scarani},
	title = {Unifying Quantum Smoothing Theories with Extended Retrodiction},
	journal = {},
	volume = {},
	pages = {},
	year = {2025},
	doi = {10.48550/arXiv.2510.08447},
	archivePrefix = {arXiv},
	eprint = {2510.08447},
	primaryClass = {quant-ph}
}

@article{Song.Kwon.Scarani:2510.26895,
	author = {M. Song, H. Kwon, and V. Scarani},
	title = {Exact and approximate conditions of tabletop reversibility: when is Petz recovery cost-free?},
	journal = {},
	volume = {},
	pages = {},
	year = {2025},
	doi = {10.48550/arXiv.2510.26895},
	archivePrefix = {arXiv},
	eprint = {2510.26895},
	primaryClass = {quant-ph}
}

@article{Chen.etal:2511.05941,
	author = {J. Chen and M. Song and J. J. X. Chan and V. Scarani},
	title = {The Petz recovery map for optical losses},
	journal = {},
	volume = {},
	pages = {},
	year = {2025},
	doi = {10.48550/arXiv.2511.05941},
	archivePrefix = {arXiv},
	eprint = {2511.05941},
	primaryClass = {quant-ph}
}

@article{Li.etal:2606.12020,
	author = {H. Li and J. Chen and Y. Pan and L. Xu and M. Song and V. Scarani and L. Zhang},
	title = {Experimental Tabletop Petz recovery of a photonic qubit},
	journal = {},
	volume = {},
	pages = {},
	year = {2026},
	doi = {10.48550/arXiv.2606.12020},
	archivePrefix = {arXiv},
	eprint = {2606.12020},
	primaryClass = {quant-ph}
}

@misc{Muller.Murk.Safranek:inprep,
    note = {F.~M{\"u}ller, S.~Murk, and D.~{\v{S}}afr{\'a}nek (in preparation).}
}

@article{Kullback.Leibler:AnnMathStat:1951,
	author = {S. Kullback and R. A. Leibler},
	title = {On Information and Sufficiency},
	journal = {Ann. Math. Stat.},
	volume = {22},
	number = {1},
	pages = {79},
	@comment = {pages={79--86}},
	year = {1951},
	doi = {10.1214/aoms/1177729694}
}

@article{Kubo.Tomita:JPhysSocJpn:1954,
	author = {R. Kubo and K. Tomita},
	title = {A General Theory of Magnetic Resonance Absorption},
	journal = {J. Phys. Soc. Jpn.},
	volume = {9},
	number = {6},
	pages = {888},
	@comment = {pages={888--919}},
	year = {1954},
	doi = {10.1143/JPSJ.9.888}
}

@article{Umegaki:KodaiMathSemRep:1962,
	author = {H. Umegaki},
	title = {Conditional expectation in an operator algebra. IV. Entropy and information},
	journal = {Kodai Math. Sem. Rep.},
	volume = {14},
	number = {2},
	pages = {59},
	@comment = {pages={59--85}},
	year = {1962},
	doi = {10.2996/kmj/1138844604}
}

@article{Mori:ProgTheorPhys:1965,
	author = {H. Mori},
	title = {Transport, Collective Motion, and Brownian Motion},
	journal = {Prog. Theor. Phys.},
	volume = {33},
	number = {3},
	pages = {423},
	@comment = {pages={423--455}},
	year = {1965},
	doi = {10.1143/PTP.33.423}
}

@article{Lieb:AdvMath:1973,
	author  = {Lieb, Elliott H.},
	title   = {Convex Trace Functions and the {Wigner--Yanase--Dyson} Conjecture},
	journal = {Adv. Math.},
	volume  = {11},
	number  = {3},
	pages   = {267},
	@comment = {pages={267--288}},
	year    = {1973},
	doi     = {10.1016/0001-8708(73)90011-X}
}

@article{Petz:CommunMathPhys:1986,
	author = {D. Petz},
	title = {Sufficient subalgebras and the relative entropy of states of a von Neumann algebra},
	journal = {Commun. Math. Phys.},
	volume = {105},
	pages = {123},
	@comment = {pages={123--131}},
	year = {1986},
	doi = {10.1007/BF01212345},
	archivePrefix = {},
	eprint = {},
	primaryClass = {}
}

@article{Petz:QuartJMathOxfordSer:1988,
	author = {D. Petz},
	title = {Sufficiency of channels over von Neumann algebras},
	journal = {Quart. J. Math. Oxford Ser.},
	volume = {39},
	number = {1},
	pages = {97},
	@comment = {pages={97--108}},
	year = {1988},
	doi = {10.1093/qmath/39.1.97}
}

@article{Vogel.Risken:PhysRevA:1989,
	author = {K. Vogel and H. Risken},
	title = {Determination of quasiprobability distributions in terms of probability distributions for the rotated quadrature phase},
	journal = {Phys. Rev. A},
	volume = {40},
	pages = {2847(R)},
	@comment = {pages={2847--2849}},
	year = {1989},
	doi = {10.1103/PhysRevA.40.2847}
}

@article{Hiai.Petz:CommunMathPhys:1991,
	author = {F. Hiai and D. Petz},
	title = {The proper formula for relative entropy and its asymptotics in quantum probability},
	journal = {Commun. Math. Phys.},
	volume = {143},
	pages = {99},
	@comment = {pages={99--114}},
	year = {1991},
	doi = {10.1007/BF02100287}
}

@article{Petz.Toth:LettMathPhys:1993,
	author = {D. Petz and G. Tóth},
	title = {The Bogoliubov inner product in quantum statistics},
	journal = {Lett. Math. Phys.},
	volume = {27},
	number = {3},
	pages = {205},
	@comment = {pages={205--216}},
	year = {1993},
	doi = {10.1007/BF00739578}
}

@article{Petz:JMathPhys:1994,
	author = {D. Petz},
	title = {Geometry of canonical correlation on the state space of a quantum system},
	journal = {J. Math. Phys.},
	volume = {35},
	pages = {780},
	@comment = {pages={780--795}},
	year = {1994},
	doi = {10.1063/1.530611},
	archivePrefix = {},
	eprint = {},
	primaryClass = {}
}

@article{Petz:LinearAlgebraAppl:1996,
	author = {D. Petz},
	title = {Monotone metrics on matrix spaces},
	journal = {Linear Algebra Appl.},
	volume = {244},
	pages = {81},
	@comment = {pages={81--96}},
	year = {1996},
	doi = {10.1016/0024-3795(94)00211-8},
	archivePrefix = {},
	eprint = {},
	primaryClass = {}
}

@article{Hradil:PhysRevA:1997,
	author = {Z. Hradil},
	title = {Quantum-state estimation},
	journal = {Phys. Rev. A},
	volume = {55},
	pages = {R1561(R)},
	year = {1997},
	doi = {10.1103/PhysRevA.55.R1561},
	archivePrefix = {arXiv},
	eprint = {quant-ph/9609012},
	primaryClass = {quant-ph},
}

@article{Chuang.Nielsen:JModOpt:1997,
	author = {I. L. Chuang and M. A. Nielsen},
	title = {Prescription for experimental determination of the dynamics of a quantum black box},
	journal = {J. Mod. Opt.},
	volume = {44},
	pages = {2455},
	@comment = {pages={2455--2467}},
	year = {1997},
	doi = {10.1080/09500349708231894},
	archivePrefix = {arXiv},
	eprint = {quant-ph/9610001},
	primaryClass = {quant-ph}
}

@article{Poyatos.Cirac.Zoller:PhysRevLett:1997,
	author = {J. F. Poyatos and J. I. Cirac and P. Zoller},
	title = {Complete Characterization of a Quantum Process: The Two-Bit Quantum Gate},
	journal = {Phys. Rev. Lett.},
	volume = {78},
	pages = {390},
	@comment = {pages={390--393}},
	year = {1997},
	doi = {10.1103/PhysRevLett.78.390},
	archivePrefix = {arXiv},
	eprint = {quant-ph/9611013},
	primaryClass = {quant-ph}
}

@article{Banaszek.etal:PhysRevA:1999,
	author = {K. Banaszek, G. M. D'Ariano, M. G. A. Paris, and M. F. Sacchi},
	title = {Maximum-likelihood estimation of the density matrix},
	journal = {Phys. Rev. A},
	volume = {61},
	pages = {010304(R)},
	year = {1999},
	doi = {10.1103/PhysRevA.61.010304},
	archivePrefix = {arXiv},
	eprint = {quant-ph/9909052},
	primaryClass = {quant-ph},
}

@article{Ogawa.Nagaoka:IEEETransInfTheory:2000,
	author = {T. Ogawa and H. Nagaoka},
	title = {Strong converse and Stein's lemma in quantum hypothesis testing},
	journal = {IEEE Trans. Inf. Theory},
	volume = {46},
	number = {7},
	pages = {2428},
	@comment = {pages={2428--2433}},
	year = {2000},
	doi = {10.1109/18.887855}
}

@article{Zyczkowski.Sommers:JPhysA:2001,
	author = {K. {\.Z}yczkowski and H.-J. Sommers},
	title = {Induced measures in the space of mixed quantum states},
	journal = {J. Phys. A: Math. Gen.},
	volume = {34},
	number = {35},
	pages = {7111},
	@comment = {pages={7111--7125}},
	year = {2001},
	doi = {10.1088/0305-4470/34/35/335},
	archivePrefix = {arXiv},
	eprint = {quant-ph/0012101},
	primaryClass = {quant-ph}
}

@article{James.etal:PhysRevA:2001,
	author = {D. F. V. James and P. G. Kwiat and W. J. Munro and A. G. White},
	title = {Measurement of qubits},
	journal = {Phys. Rev. A},
	volume = {64},
	pages = {052312},
	year = {2001},
	doi = {10.1103/PhysRevA.64.052312},
	archivePrefix = {arXiv},
	eprint = {quant-ph/0103121},
	primaryClass = {quant-ph}
}

@article{Barnum.Knill:JMathPhys:2002,
	author = {H. Barnum and E. Knill},
	title = {Reversing quantum dynamics with near-optimal quantum and classical fidelity},
	journal = {J. Math. Phys.},
	volume = {43},
	number = {5},
	pages = {2097},
	@comment = {pages={2097--2106}},
	year = {2002},
	doi = {10.1063/1.1459754},
	archivePrefix = {arXiv},
	eprint = {quant-ph/0004088},
	primaryClass = {quant-ph}
}

@article{Ruskai:JMathPhys:2002,
	author = {M. B. Ruskai},
	title = {Inequalities for Quantum Entropy: A Review with Conditions for Equality},
	journal = {J. Math. Phys.},
	volume = {43},
	number = {9},
	pages = {4358},
	@comment = {pages={4358--4375}},
	year = {2002},
	doi = {10.1063/1.1497701},
	archivePrefix = {arXiv},
	eprint = {quant-ph/0205064},
	primaryClass = {quant-ph},
	note = {Erratum: \href{https://doi.org/10.1063/1.1824214}{J. Math. Phys. \textbf{46}, 019901 (2005)}}
}

@article{Jezek.Fiurasek.Hradil:PhysRevA:2003,
	author = {M. Ježek, J. Fiurášek, and Z. Hradil},
	title = {Quantum inference of states and processes},
	journal = {Phys. Rev. A},
	volume = {68},
	pages = {012305},
	year = {2003},
	doi = {10.1103/PhysRevA.68.012305},
	archivePrefix = {arXiv},
	eprint = {quant-ph/0210146},
	primaryClass = {quant-ph}
}

@article{Petz:RevMathPhys:2003,
	author = {D. Petz},
	title = {Monotonicity of quantum relative entropy revisited},
	journal = {Rev. Math. Phys.},
	volume = {15},
	number = {1},
	pages = {79},
	@comment = {pages={79--91}},
	year = {2003},
	doi = {10.1142/S0129055X03001576},
	archivePrefix = {arXiv},
	eprint = {quant-ph/0209053},
	primaryClass = {quant-ph}
}

@article{DAriano.Paris.Sacchi:AdvImagingElectronPhys:2003,
	author = {G. M. D'Ariano and M. G. A. Paris and M. F. Sacchi},
	title = {Quantum Tomography},
	journal = {Adv. Imaging Electron Phys.},
	volume = {128},
	pages = {205},
	@comment = {pages={205--308}},
	year = {2003},
	doi = {10.1016/S1076-5670(03)80065-4},
	archivePrefix = {arXiv},
	eprint = {quant-ph/0302028},
	primaryClass = {quant-ph}
}

@article{Hayden.etal:CommunMathPhys:2004,
	author = {P. Hayden, R. Jozsa, D. Petz, and A. Winter},
	title = {Structure of States Which Satisfy Strong Subadditivity of Quantum Entropy with Equality},
	journal = {Commun. Math. Phys.},
	volume = {246},
	pages = {359},
	@comment = {pages={359--374}},
	year = {2004},
	doi = {10.1007/s00220-004-1049-z},
	archivePrefix = {arXiv},
	eprint = {quant-ph/0304007},
	primaryClass = {quant-ph}
}

@article{Chan.Darwiche:ArtifIntell:2005,
	author = {H. Chan and A. Darwiche},
	title = {On the revision of probabilistic beliefs using uncertain evidence},
	journal = {Artif. Intell.},
	volume = {163},
	number = {1},
	pages = {67},
	@comment = {pages={67-90}},
	year = {2005},
	doi = {10.1016/j.artint.2004.09.005}
}

@article{Bennett.etal:PhysRevA:2006,
	author = {C. H. Bennett and P. W. Shor and J. A. Smolin and A. V. Thapliyal},
	title = {Universal quantum data compression via nondestructive tomography},
	journal = {Phys. Rev. A},
	volume = {73},
	pages = {032336},
	year = {2006},
	doi = {10.1103/PhysRevA.73.032336},
	archivePrefix = {arXiv},
	eprint = {quant-ph/0403078},
	primaryClass = {quant-ph}
}

@article{Rehacek.etal:PhysRevA:2007,
	author = {J. Řeháček, Z. Hradil, E. Knill, A. I. Lvovsky},
	title = {Diluted maximum-likelihood algorithm for quantum tomography},
	journal = {Phys. Rev. A},
	volume = {75},
	pages = {042108},
	year = {2007},
	doi = {10.1103/PhysRevA.75.042108},
	archivePrefix = {arXiv},
	eprint = {quant-ph/0611244},
	primaryClass = {quant-ph}
}

@article{Blume-Kohout:NewJPhys:2010,
	author = {R. Blume-Kohout},
	title = {Optimal, reliable estimation of quantum states},
	journal = {New J. Phys.},
	volume = {12},
	number = {4},
	pages = {043034},
	year = {2010},
	doi = {10.1088/1367-2630/12/4/043034},
	archivePrefix = {arXiv},
	eprint = {quant-ph/0611080},
	primaryClass = {quant-ph}
}

@article{Gross.etal:PhysRevLett:2010,
	author = {D. Gross, Y.-K. Liu, S. T. Flammia, S. Becker, and J. Eisert},
	title = {Quantum State Tomography via Compressed Sensing},
	journal = {Phys. Rev. Lett.},
	volume = {105},
	pages = {150401},
	year = {2010},
	doi = {10.1103/PhysRevLett.105.150401},
	archivePrefix = {arXiv},
	eprint = {0909.3304},
	primaryClass = {quant-ph}
}

@article{Ng.Mandayam:PhysRevA:2010,
	author = {H. K. Ng and P. Mandayam},
	title = {Simple approach to approximate quantum error correction based on the transpose channel},
	journal = {Phys. Rev. A},
	volume = {81},
	pages = {062342},
	year = {2010},
	doi = {10.1103/PhysRevA.81.062342},
	archivePrefix = {arXiv},
	eprint = {0909.0931},
	primaryClass = {quant-ph}
}

@article{BlumeKohout:PhysRevLett:2010,
	author = {R. Blume-Kohout},
	title = {Hedged Maximum Likelihood Quantum State Estimation},
	journal = {Phys. Rev. Lett.},
	volume = {105},
	pages = {200504},
	year = {2010},
	doi = {10.1103/PhysRevLett.105.200504},
	archivePrefix = {arXiv},
	eprint = {1001.2029},
	primaryClass = {quant-ph}
}

@article{Cramer.etal:NatCommun:2010,
	author = {M. Cramer and M. B. Plenio and S. T. Flammia and D. Gross and S. D. Bartlett and R. Somma and O. Landon-Cardinal and Y.-K. Liu and D. Poulin},
	title = {Efficient quantum state tomography},
	journal = {Nat. Commun.},
	volume = {1},
	pages = {149},
	year = {2010},
	doi = {10.1038/ncomms1147},
	archivePrefix = {arXiv},
	eprint = {1101.4366},
	primaryClass = {quant-ph}
}

@article{HuszarHoulsby:PhysRevA:2012,
	author = {F. Huszár and N. M. T. Houlsby},
	title = {Adaptive Bayesian quantum tomography},
	journal = {Phys. Rev. A},
	volume = {85},
	issue = {5},
	pages = {052120},
	year = {2012},
	doi = {10.1103/PhysRevA.85.052120},
	archivePrefix = {arXiv},
	eprint = {1107.0895},
	primaryClass = {quant-ph}
}

@article{Flammia.etal:NewJPhys:2012,
	author = {S. T. Flammia and D. Gross and Y.-K. Liu and J. Eisert},
	title = {Quantum tomography via compressed sensing: error bounds, sample complexity and efficient estimators},
	journal = {New J. Phys.},
	volume = {14},
	pages = {095022},
	year = {2012},
	doi = {10.1088/1367-2630/14/9/095022},
	archivePrefix = {arXiv},
	eprint = {1205.2300},
	primaryClass = {quant-ph}
}

@article{Leifer.Spekkens:PhysRevA:2013,
	author = {M. S. Leifer and R. W. Spekkens},
	title = {Towards a formulation of quantum theory as a causally neutral theory of Bayesian inference},
	journal = {Phys. Rev. A},
	volume = {88},
	pages = {052130},
	year = {2013},
	doi = {10.1103/PhysRevA.88.052130},
	archivePrefix = {arXiv},
	eprint = {1107.5849},
	primaryClass = {quant-ph}
}

@article{Wu:SciRep:2013,
	author = {S. Wu},
	title = {State tomography via weak measurements},
	journal = {Sci. Rep.},
	volume = {3},
	pages = {1193},
	year = {2013},
	doi = {10.1038/srep01193},
	archivePrefix = {arXiv},
	eprint = {1212.3655},
	primaryClass = {quant-ph}
}

@article{Goncalves:QuantumInfComput:2013,
	author = {D. S. Gon{\c{c}}alves and M. A. Gomes-Ruggiero and C. Lavor},
	title = {Global convergence of diluted iterations in maximum-likelihood quantum tomography},
	journal = {Quantum Inf. Comput.},
	volume = {14},
	number = {11--12},
	pages = {966},
	@comment = {pages={966--980}},
	year = {2014},
	doi = {10.26421/QIC14.11-12-5},
	archivePrefix = {arXiv},
	eprint = {1306.3057},
	primaryClass = {math-ph}
}

@article{Ferrie:PhysRevLett:2014,
	author = {C. Ferrie},
	title = {Self-Guided Quantum Tomography},
	journal = {Phys. Rev. Lett.},
	volume = {113},
	pages = {190404},
	year = {2014},
	doi = {10.1103/PhysRevLett.113.190404},
	archivePrefix = {arXiv},
	eprint = {1406.4101},
	primaryClass = {quant-ph}
}

@article{Fawzi.Renner:CommunMathPhys:2015,
	author = {O. Fawzi and R. Renner},
	title = {Quantum Conditional Mutual Information and Approximate Markov Chains},
	journal = {Commun. Math. Phys.},
	volume = {340},
	pages = {575},
	@comment = {pages={575--611}},
	year = {2015},
	doi = {10.1007/s00220-015-2466-x},
	archivePrefix = {arXiv},
	eprint = {1410.0664},
	primaryClass = {quant-ph}
}

@article{Kueng.Ferrie:NewJPhys:2015,
	author = {R. Kueng and C. Ferrie},
	title = {Near-optimal quantum tomography: estimators and bounds},
	journal = {New J. Phys.},
	volume = {17},
	number = {12},
	pages = {123013},
	year = {2015},
	doi = {10.1088/1367-2630/17/12/123013},
	archivePrefix = {arXiv},
	eprint = {1503.00677},
	primaryClass = {quant-ph}
}

@article{Wilde:ProcRSocA:2015,
	author = {M. M. Wilde},
	title = {Recoverability in quantum information theory},
	journal = {Proc. R. Soc. A},
	volume = {471},
	pages = {20150338},
	year = {2015},
	doi = {10.1098/rspa.2015.0338},
	archivePrefix = {arXiv},
	eprint = {1505.04661},
	primaryClass = {quant-ph}
}

@article{Beigi.Datta.Leditzky:JMathPhys:2016,
	author = {S. Beigi and N. Datta and F. Leditzky},
	title = {Decoding quantum information via the Petz recovery map},
	journal = {J. Math. Phys.},
	volume = {57},
	pages = {082203},
	year = {2016},
	doi = {10.1063/1.4961515},
	archivePrefix = {arXiv},
	eprint = {1504.04449},
	primaryClass = {quant-ph}
}

@article{Sutter.Fawzi.Renner:ProcRSocA:2016,
	author = {D. Sutter, O. Fawzi, and R. Renner},
	title = {Universal recovery map for approximate Markov chains},
	journal = {Proc. R. Soc. A},
	volume = {472},
	pages = {20150623},
	year = {2016},
	doi = {10.1098/rspa.2015.0623},
	archivePrefix = {arXiv},
	eprint = {1504.07251},
	primaryClass = {quant-ph}
}

@article{Sutter.Tomamichel.Harrow:IEEETransInfTheory:2016,
	author = {D. Sutter, M. Tomamichel, and A. W. Harrow},
	title = {Strengthened Monotonicity of Relative Entropy via Pinched Petz Recovery Map},
	journal = {IEEE Trans. Inf. Theory},
	volume = {62},
	number = {5},
	pages = {2907},
	@comment = {pages={2907--2913}},
	year = {2016},
	doi = {10.1109/TIT.2016.2545680},
	archivePrefix = {arXiv},
	eprint = {1507.00303},
	primaryClass = {quant-ph}
}

@article{Granade.Combes.Cory:NewJPhys:2016,
	author = {C. Granade and J. Combes and D. G. Cory},
	title = {Practical Bayesian tomography},
	journal = {New J. Phys.},
	volume = {18},
	pages = {033024},
	year = {2016},
	doi = {10.1088/1367-2630/18/3/033024},
	archivePrefix = {arXiv},
	eprint = {1509.03770},
	primaryClass = {quant-ph}
}

@article{Struchalin.etal:PhysRevA:2016,
	author = {G. I. Struchalin and I. A. Pogorelov and S. S. Straupe and K. S. Kravtsov and I. V. Radchenko and S. P. Kulik},
	title = {Experimental adaptive quantum tomography of two-qubit states},
	journal = {Phys. Rev. A},
	volume = {93},
	pages = {012103},
	year = {2016},
	doi = {10.1103/PhysRevA.93.012103},
	archivePrefix = {arXiv},
	eprint = {1510.05303},
	primaryClass = {quant-ph}
}

@article{Chapman.Ferrie.Peruzzo:PhysRevLett:2016,
	author = {R. J. Chapman and C. Ferrie and A. Peruzzo},
	title = {Experimental Demonstration of Self-Guided Quantum Tomography},
	journal = {Phys. Rev. Lett.},
	volume = {117},
	pages = {040402},
	year = {2016},
	doi = {10.1103/PhysRevLett.117.040402},
	archivePrefix = {arXiv},
	eprint = {1602.04194},
	primaryClass = {quant-ph}
}

@article{Straupe:JetpLett:2016,
	author = {S. S. Straupe},
	title = {Adaptive quantum tomography},
	journal = {JETP Lett.},
	volume = {104},
	pages = {510},
	@comment = {pages={510-522}},
	year = {2016},
	doi = {10.1134/S0021364016190024},
	archivePrefix = {arXiv},
	eprint = {1610.02840},
	primaryClass = {quant-ph}
}

@article{Sutter.Berta.Tomamichel:CommunMathPhys:2017,
	author = {D. Sutter and M. Berta and M. Tomamichel},
	title = {Multivariate Trace Inequalities},
	journal = {Commun. Math. Phys.},
	volume = {352},
	pages = {37},
	@comment = {pages={37--58}},
	year = {2017},
	doi = {10.1007/s00220-016-2778-5},
	archivePrefix = {arXiv},
	eprint = {1604.03023},
	primaryClass = {math-ph}
}

@article{Lanyon.etal:NatPhys:2017,
	author = {B. P. Lanyon and C. Maier and M. Holz{\"a}pfel and T. Baumgratz and C. Hempel and P. Jurcevic and I. Dhand and A. S. Buyskikh and A. J. Daley and M. Cramer and M. B. Plenio and R. Blatt and C. F. Roos},
	title = {Efficient tomography of a quantum many-body system},
	journal = {Nat. Phys.},
	volume = {13},
	pages = {1158},
	@comment = {pages={1158--1162}},
	year = {2017},
	doi = {10.1038/nphys4244},
	archivePrefix = {arXiv},
	eprint = {1612.08000},
	primaryClass = {quant-ph}
}

@article{Bolduc.etal:NPJQuantumInf:2017,
	author = {E. Bolduc and G. C. Knee and E. M. Gauger and J. Leach},
	title = {Projected gradient descent algorithms for quantum state tomography},
	journal = {npj Quantum Inf.},
	volume = {3},
	pages = {44},
	year = {2017},
	doi = {10.1038/s41534-017-0043-1},
	archivePrefix = {arXiv},
	eprint = {1612.09531},
	primaryClass = {quant-ph}
}

@article{Wilming.Gallego.Eisert:Entropy:2017,
    author = {H. Wilming and R. Gallego and J. Eisert},
    title = {Axiomatic Characterization of the Quantum Relative Entropy and Free Energy},
    journal = {Entropy},
    volume = {19},
    pages = {241},
    year = {2017},
    doi = {10.3390/e19060241},
    archivePrefix = {arxiv},
    eprint = {1702.08473},
    primaryClass = {quant-ph}
}

@article{Junge.etal:AnnHenriPoincare:2018,
	author = {M. Junge, R. Renner, D. Sutter, M. M. Wilde, and A. Winter},
	title = {Universal Recovery Maps and Approximate Sufficiency of Quantum Relative Entropy},
	journal = {Ann. Henri Poincaré},
	volume = {19},
	pages = {2955},
	@comment = {pages={2955--2978}},
	year = {2018},
	doi = {10.1007/s00023-018-0716-0},
	archivePrefix = {arXiv},
	eprint = {1509.07127},
	primaryClass = {quant-ph}
}

@article{Aberg:PhysRevX:2018,
	author = {J. {\AA}berg},
	title = {Fully Quantum Fluctuation Theorems},
	journal = {Phys. Rev. X},
	volume = {8},
	pages = {011019},
	year = {2018},
	doi = {10.1103/PhysRevX.8.011019},
	archivePrefix = {arXiv},
	eprint = {1601.01302},
	primaryClass = {quant-ph}
}

@article{Torlai.etal:NatPhys:2018,
	author = {G. Torlai and G. Mazzola and J. Carrasquilla and M. Troyer and R. Melko and G. Carleo},
	title = {Neural-network quantum state tomography},
	journal = {Nat. Phys.},
	volume = {14},
	pages = {447},
	@comment = {pages={447--450}},
	year = {2018},
	doi = {10.1038/s41567-018-0048-5},
	archivePrefix = {arXiv},
	eprint = {1703.05334},
	primaryClass = {cond-mat.dis-nn}
}

@article{Cotler.etal:PhysRevX:2019,
	author = {J. Cotler and P. Hayden and G. Penington and G. Salton and B. Swingle and M. Walter},
	title = {Entanglement Wedge Reconstruction via Universal Recovery Channels},
	journal = {Phys. Rev. X},
	volume = {9},
	pages = {031011},
	year = {2019},
	doi = {10.1103/PhysRevX.9.031011},
	archivePrefix = {arXiv},
	eprint = {1704.05839},
	primaryClass = {hep-th}
}

@article{Xin.etal:NPJQuantumInf:2019,
	author = {T. Xin and S. Lu and N. Cao and G. Anikeeva and D. Lu and J. Li and G. Long and B. Zeng},
	title = {Local-measurement-based quantum state tomography via neural networks},
	journal = {npj Quantum Inf.},
	volume = {5},
	pages = {109},
	year = {2019},
	doi = {10.1038/s41534-019-0222-3},
	archivePrefix = {arXiv},
	eprint = {1807.07445},
	primaryClass = {quant-ph}
}

@article{Kwon.Kim:PhysRevX:2019,
	author = {H. Kwon and M. S. Kim},
	title = {Fluctuation Theorems for a Quantum Channel},
	journal = {Phys. Rev. X},
	volume = {9},
	pages = {031029},
	year = {2019},
	doi = {10.1103/PhysRevX.9.031029},
	archivePrefix = {arXiv},
	eprint = {1810.03150},
	primaryClass = {quant-ph}
}

@article{Carrasquilla.etal:NatMachIntell:2019,
	author = {J. Carrasquilla and G. Torlai and R. G. Melko and L. Aolita},
	title = {Reconstructing quantum states with generative models},
	journal = {Nat. Mach. Intell.},
	volume = {1},
	pages = {155},
	@comment = {pages={155--161}},
	year = {2019},
	doi = {10.1038/s42256-019-0028-1},
	archivePrefix = {arXiv},
	eprint = {1810.10584},
	primaryClass = {quant-ph}
}

@article{Guta.etal:JPhysA:2020,
	author = {M. Guţă, J. Kahn, R. Kueng, and J. A. Tropp},
	title = {Fast state tomography with optimal error bounds},
	journal = {J. Phys. A: Math. Theor.},
	volume = {53},
	pages = {204001},
	year = {2020},
	doi = {10.1088/1751-8121/ab8111},
	archivePrefix = {arXiv},
	eprint = {1809.11162},
	primaryClass = {quant-ph}
}

@article{Chen.Penington.Salton:JHEP:2020,
	author = {C.-F. Chen and G. Penington and G. Salton},
	title = {Entanglement wedge reconstruction using the Petz map},
	journal = {J. High Energy Phys.},
	volume = {01},
	pages = {168},
	year = {2020},
	doi = {10.1007/JHEP01(2020)168},
	archivePrefix = {arXiv},
	eprint = {1902.02844},
	primaryClass = {hep-th}
}

@article{Palmieri.etal:NPJQuantumInf:2020,
	author = {A. M. Palmieri and E. Kovlakov and F. Bianchi and D. Yudin and S. Straupe and J. D. Biamonte and S. Kulik},
	title = {Experimental neural network enhanced quantum tomography},
	journal = {npj Quantum Inf.},
	volume = {6},
	pages = {20},
	year = {2020},
	doi = {10.1038/s41534-020-0248-6},
	archivePrefix = {arXiv},
	eprint = {1904.05902},
	primaryClass = {quant-ph}
}

@article{Huang.Kueng.Preskill:NatPhys:2020,
	author = {H.-Y. Huang and R. Kueng and J. Preskill},
	title = {Predicting many properties of a quantum system from very few measurements},
	journal = {Nat. Phys.},
	volume = {16},
	pages = {1050},
	@comment = {pages={1050--1057}},
	year = {2020},
	doi = {10.1038/s41567-020-0932-7},
	archivePrefix = {arXiv},
	eprint = {2002.08953},
	primaryClass = {quant-ph}
}

@article{Lukens.etal:NewJPhys:2020,
	author = {J. M. Lukens and K. J. H. Law and A. Jasra and P. Lougovski},
	title = {A practical and efficient approach for Bayesian quantum state estimation},
	journal = {New J. Phys.},
	volume = {22},
	pages = {063038},
	year = {2020},
	doi = {10.1088/1367-2630/ab8efa},
	archivePrefix = {arXiv},
	eprint = {2002.10354},
	primaryClass = {quant-ph}
}

@article{Lohani.etal:MachLearnSciTechnol:2020,
	author = {S. Lohani and B. T. Kirby and M. Brodsky and O. Danaci and R. T. Glasser},
	title = {Machine learning assisted quantum state estimation},
	journal = {Mach. Learn.: Sci. Technol.},
	volume = {1},
	pages = {035007},
	year = {2020},
	doi = {10.1088/2632-2153/ab9a21},
	archivePrefix = {arXiv},
	eprint = {2003.03441},
	primaryClass = {quant-ph}
}

@article{Buscemi.Scarani:PhysRevE:2021,
	author = {F. Buscemi and V. Scarani},
	title = {Fluctuation theorems from Bayesian retrodiction},
	journal = {Phys. Rev. E},
	volume = {103},
	pages = {052111},
	year = {2021},
	doi = {10.1103/PhysRevE.103.052111},
	archivePrefix = {arXiv},
	eprint = {2009.02849},
	primaryClass = {quant-ph}
}

@article{Aw.Buscemi.Scarani:AVSQuantumSci:2021,
	author = {C. C. Aw and F. Buscemi and V. Scarani},
	title = {Fluctuation theorems with retrodiction rather than reverse processes},
	journal = {AVS Quantum Sci.},
	volume = {3},
	pages = {045601},
	year = {2021},
	doi = {10.1116/5.0060893},
	archivePrefix = {arXiv},
	eprint = {2106.08589},
	primaryClass = {cond-mat.stat-mech}
}

@article{Jacobs:ElectronProcTheorComputSci.2021,
	author = {B. Jacobs},
	title = {Learning from What's Right and Learning from What's Wrong},
	journal = {Electron. Proc. Theor. Comput. Sci.},
	volume = {351},
	pages = {116},
	@comment = {pages={116-133}},
	year = {2021},
	doi = {10.4204/EPTCS.351.8},
	archivePrefix = {arXiv},
	eprint = {2112.14045},
	primaryClass = {cs.LO}
}

@article{Tsang:PhysRevA:2022,
	author = {M. Tsang},
	title = {Generalized conditional expectations for quantum retrodiction and smoothing},
	journal = {Phys. Rev. A},
	volume = {105},
	pages = {042213},
	year = {2022},
	doi = {10.1103/PhysRevA.105.042213},
	archivePrefix = {arXiv},
	eprint = {1912.02711},
	primaryClass = {quant-ph}
}

@article{Parzygnat.Russo:LinearAlgebrAppl:2022,
	author = {A. J. Parzygnat and B. P. Russo},
	title = {A non-commutative Bayes' theorem},
	journal = {Linear Algebr. Appl.},
	volume = {644},
	pages = {28},
	@comment = {pages={28-94}},
	year = {2022},
	doi = {10.1016/j.laa.2022.02.030},
	archivePrefix = {arXiv},
	eprint = {2005.03886},
	primaryClass = {quant-ph}
}

@article{Gilyen.etal:PhysRevLett:2022,
	author = {A. Gilyén, S. Lloyd, I. Marvian, Y. Quek, and M. M. Wilde},
	title = {Quantum Algorithm for Petz Recovery Channels and Pretty Good Measurements},
	journal = {Phys. Rev. Lett.},
	volume = {128},
	pages = {220502},
	year = {2022},
	doi = {10.1103/PhysRevLett.128.220502},
	archivePrefix = {arXiv},
	eprint = {2006.16924},
	primaryClass = {quant-ph}
}

@article{Chapman.etal:OptExpress:2022,
	author = {J. C. Chapman and J. M. Lukens and B. Qi and R. C. Pooser and N. A. Peters},
	title = {Bayesian homodyne and heterodyne tomography},
	journal = {Opt. Express},
	volume = {30},
	pages = {15184},
	@comment = {pages={15184--15200}},
	year = {2022},
	doi = {10.1364/OE.456597},
	archivePrefix = {arXiv},
	eprint = {2202.03499},
	primaryClass = {quant-ph}
}

@article{Stricker.etal:PRXQuantum:2022,
	author = {R. Stricker and M. Meth and L. Postler and C. Edmunds and C. Ferrie and R. Blatt and P. Schindler and T. Monz and R. Kueng and M. Ringbauer},
	title = {Experimental Single-Setting Quantum State Tomography},
	journal = {PRX Quantum},
	volume = {3},
	pages = {040310},
	year = {2022},
	doi = {10.1103/PRXQuantum.3.040310},
	archivePrefix = {arXiv},
	eprint = {2206.00019},
	primaryClass = {quant-ph}
}

@article{Junge:JStatPhys:2023,
	author = {M. Junge and N. LaRacuente and C. Rouz{\'e}},
	title = {Stability of Logarithmic Sobolev Inequalities Under a Noncommutative Change of Measure},
	journal = {J. Stat. Phys.},
	volume = {190},
	number = {2},
	pages = {30},
	year = {2023},
	doi = {10.1007/s10955-022-03026-x},
	archivePrefix = {arXiv},
	eprint = {1911.08533},
	primaryClass = {quant-ph}
}

@article{Surace.Scandi:Quantum:2023,
	author = {J. Surace and M. Scandi},
	title = {State retrieval beyond Bayes' retrodiction},
	journal = {Quantum},
	volume = {7},
	pages = {990},
	year = {2023},
	doi = {10.22331/q-2023-04-27-990},
	archivePrefix = {arXiv},
	eprint = {2201.09899},
	primaryClass = {quant-ph}
}

@article{Parzygnat.Buscemi:Quantum:2023,
	author = {A. J. Parzygnat and F. Buscemi},
	title = {Axioms for retrodiction: achieving time-reversal symmetry with a prior},
	journal = {Quantum},
	volume = {7},
	pages = {1013},
	year = {2023},
	doi = {10.22331/q-2023-05-23-1013},
	archivePrefix = {arXiv},
	eprint = {2210.13531},
	primaryClass = {quant-ph}
}

@article{Buscemi.Schindler.Safranek:NewJPhys:2023,
	author = {F. Buscemi, J. Schindler, and D. Šafránek},
	title = {Observational entropy, coarse-grained states, and the Petz recovery map: information-theoretic properties and bounds},
	journal = {New J. Phys.},
	volume = {25},
	pages = {053002},
	year = {2023},
	doi = {10.1088/1367-2630/accd11},
	archivePrefix = {arXiv},
	eprint = {2209.03803},
	primaryClass = {quant-ph}
}

@article{Parzygnat.Fullwood:PRXQuantum:2023,
	author = {A. J. Parzygnat and J. Fullwood},
	title = {From Time-Reversal Symmetry to Quantum Bayes' Rules},
	journal = {PRX Quantum},
	volume = {4},
	pages = {020334},
	year = {2023},
	doi = {10.1103/PRXQuantum.4.020334},
	archivePrefix = {arXiv},
	eprint = {2212.08088},
	primaryClass = {quant-ph}
}

@article{Aw.etal:PRXQuantum:2023,
	author = {C. C. Aw and K. Onggadinata and D. Kaszlikowski and V. Scarani},
	title = {Quantum Bayesian Inference in Quasiprobability Representations},
	journal = {PRX Quantum},
	volume = {4},
	pages = {020352},
	year = {2023},
	doi = {10.1103/PRXQuantum.4.020352},
	archivePrefix = {arXiv},
	eprint = {2301.01952},
	primaryClass = {quant-ph}
}

@article{Vernooij.Wirth:CommunMathPhys:2023,
	author = {M. Vernooij and M. Wirth},
	title = {Derivations and KMS-Symmetric Quantum Markov Semigroups},
	journal = {Commun. Math. Phys.},
	volume = {403},
	pages = {381},
	@comment = {pages={381-416}},
	year = {2023},
	doi = {10.1007/s00220-023-04795-6},
	archivePrefix = {arXiv},
	eprint = {2303.15949},
	primaryClass = {math.OA}
}

@article{Aw.etal:PRXQuantum:2024,
	author = {C. C. Aw, L. H. Zaw, M. Balanzó-Juandó, and V. Scarani},
	title = {Role of Dilations in Reversing Physical Processes: Tabletop Reversibility and Generalized Thermal Operations},
	journal = {PRX Quantum},
	volume = {5},
	pages = {010332},
	year = {2024},
	doi = {10.1103/PRXQuantum.5.010332},
	archivePrefix = {arXiv},
	eprint = {2308.13909},
	primaryClass = {quant-ph}
}

@article{Zheng.etal:PhysRevLett:2024,
	author = {G. Zheng and W. He and G. Lee and L. Jiang},
	title = {Near-Optimal Performance of Quantum Error Correction Codes},
	journal = {Phys. Rev. Lett.},
	volume = {132},
	pages = {250602},
	year = {2024},
	doi = {10.1103/PhysRevLett.132.250602},
	archivePrefix = {arXiv},
	eprint = {2401.02022},
	primaryClass = {quant-ph}
}

@article{Scandi.etal:RepProgPhys:2025,
	author = {M. Scandi and P. Abiuso and J. Surace and D. De Santis},
	title = {Quantum Fisher information and its dynamical nature},
	journal = {Rep. Prog. Phys.},
	volume = {88},
	pages = {076001},
	year = {2025},
	doi = {10.1088/1361-6633/ade453},
	archivePrefix = {arXiv},
	eprint = {2304.14984},
	primaryClass = {quant-ph}
}

@article{Qin.etal:NPJQuantumInf:2025,
	author = {Z. Qin and J. M. Lukens and B. T. Kirby and Z. Zhu},
	title = {Enhancing quantum state reconstruction with structured classical shadows},
	journal = {npj Quantum Inf.},
	volume = {11},
	pages = {147},
	year = {2025},
	doi = {10.1038/s41534-025-01101-1},
	archivePrefix = {arXiv},
	eprint = {2501.03144},
	primaryClass = {quant-ph}
}

@article{Liu.etal:PhysRevA:2025,
	author = {M. Liu and V. Scarani and A. Auff{\`e}ves and K. T. Laverick},
	title = {Retrodictive approach to quantum state smoothing},
	journal = {Phys. Rev. A},
	volume = {112},
	pages = {L030203},
	year = {2025},
	doi = {10.1103/8pc3-7pg5},
	archivePrefix = {arXiv},
	eprint = {2501.15986},
	primaryClass = {quant-ph}
}

@article{Png.Scarani:PhysRevA:2025,
	author = {W.-H. Png and V. Scarani},
	title = {Petz recovery maps of single-qubit decoherence channels in an ion trap quantum processor},
	journal = {Phys. Rev. A},
	volume = {112},
	pages = {022613},
	year = {2025},
	doi = {10.1103/7f8x-n2np},
	archivePrefix = {arXiv},
	eprint = {2504.20399},
	primaryClass = {quant-ph}
}

@article{Bai.Buscemi.Scarani:PhysRevLett:2025,
	author = {G. Bai, F. Buscemi, and V. Scarani},
	title = {Quantum Bayes’ Rule and Petz Transpose Map from the Minimum Change Principle},
	journal = {Phys. Rev. Lett.},
	volume = {135},
	pages = {090203},
	year = {2025},
	doi = {10.1103/5n4p-bxhm},
	archivePrefix = {arXiv},
	eprint = {2410.00319},
	primaryClass = {quant-ph}
}

@article{Gaikwad.etal:QuantumSciTechnol:2025,
	author = {A. Gaikwad, M. S. Torres, S. Ahmed, and A. F. Kockum},
	title = {Gradient-descent methods for fast quantum state tomography},
	journal = {Quantum Sci. Technol.},
	volume = {10},
	pages = {045055},
	year = {2025},
	doi = {10.1088/2058-9565/ae0baa},
	archivePrefix = {arXiv},
	eprint = {2503.04526},
	primaryClass = {quant-ph}
}

@article{Aditi.Becker:PhysRevA:2025,
	author = {K. Aditi and S. Becker},
	title = {Rigorous maximum-likelihood estimation for quantum states},
	journal = {Phys. Rev. A},
	volume = {112},
	pages = {052436},
	year = {2025},
	doi = {10.1103/j5gh-hmtw},
	archivePrefix = {arXiv},
	eprint = {2506.16646},
	primaryClass = {quant-ph}
}

@article{Nagasawa.etal:RepProgPhys:2025,
	author = {T. Nagasawa and E. Wakakuwa and K. Kato and F. Buscemi},
	title = {Macroscopicity and observational deficit in states, operations, and correlations},
	journal = {Rep. Prog. Phys.},
	volume = {88},
	number = {11},
	pages = {117601},
	year = {2025},
	doi = {10.1088/1361-6633/ae140e}
}

@article{Liu.Scarani.Bai:PhysRevLett:2026,
	author = {M. Liu and V. Scarani and G. Bai},
	title = {Proper and Improper Mixed States Serve as Different Prior Beliefs for Quantum State Retrodiction},
	journal = {Phys. Rev. Lett.},
	volume = {136},
	pages = {060203},
	year = {2026},
	doi = {10.1103/xx43-p1py},
	archivePrefix = {arXiv},
	eprint = {2502.10030},
	primaryClass = {quant-ph}
}

@article{Singh.etal:PhysRevA:2026,
	author = {G. Singh and R. S. Sahani and V. Jagadish and L. Lautenbacher and N. K. Bernardes and K. Dorai},
	title = {Realizing the Petz recovery map on an NMR quantum processor},
	journal = {Phys. Rev. A},
	volume = {113},
	pages = {052415},
	year = {2026},
	doi = {10.1103/xd6k-swv7},
	archivePrefix = {arXiv},
	eprint = {2508.08998},
	primaryClass = {quant-ph}
}

@article{Kuang.Torii.Buscemi:QuantumSciTechnol:2026,
	author = {J. Kuang and K. Torii and F. Buscemi},
	title = {Quantum measurement retrodiction and entropic uncertainty relations},
	journal = {Quantum Sci. Technol.},
	volume = {11},
	number = {2},
	pages = {025049},
	year = {2026},
	doi = {10.1088/2058-9565/ae5fcd},
    archivePrefix = {arXiv},
    eprint = {2511.20281},
    primaryClass = {quant-ph}
}

@article{Lvovsky.Raymer:RevModPhys:2009,
	author = {A. I. Lvovsky and M. G. Raymer},
	title = {Continuous-variable optical quantum-state tomography},
	journal = {Rev. Mod. Phys.},
	volume = {81},
	pages = {299},
	@comment = {pages={299--332}},
	year = {2009},
	doi = {10.1103/RevModPhys.81.299},
	archivePrefix = {arXiv},
	eprint = {quant-ph/0511044},
	primaryClass = {quant-ph}
}

@article{Scarani.etal:RevModPhys:2009,
	author = {V. Scarani and H. Bechmann-Pasquinucci and N. J. Cerf and M. Du{\v{s}}ek and N. L{\"u}tkenhaus and M. Peev},
	title = {The security of practical quantum key distribution},
	journal = {Rev. Mod. Phys.},
	volume = {81},
	pages = {1301},
	@comment = {pages={1301--1350}},
	year = {2009},
	doi = {10.1103/RevModPhys.81.1301},
	archivePrefix = {arXiv},
	eprint = {0802.4155},
	primaryClass = {quant-ph}
}

@article{Degen.Reinhard.Cappellaro:RevModPhys:2017,
	author = {C. L. Degen and F. Reinhard and P. Cappellaro},
	title = {Quantum sensing},
	journal = {Rev. Mod. Phys.},
	volume = {89},
	pages = {035002},
	year = {2017},
	doi = {10.1103/RevModPhys.89.035002},
	archivePrefix = {arXiv},
	eprint = {1611.02427},
	primaryClass = {quant-ph}
}

@book{Pearl:Book:1988,
	author = {J. Pearl},
	title = {Probabilistic Reasoning in Intelligent Systems: Networks of Plausible Inference},
	publisher = {Morgan Kaufmann Publishers},
	address = {San Mateo, CA},
	year = {1988},
	doi = {10.1016/C2009-0-27609-4}
}

@book{Jeffrey:Book:1990,
	author = {R. C. Jeffrey},
	title = {The Logic of Decision},
	publisher = {The University of Chicago Press},
	address = {Chicago, IL},
	year = {1990},
	edition = {2},
	doi = {}
}

@book{Leonhardt:Book:1997,
	author = {U. Leonhardt},
	title = {Measuring the Quantum State of Light},
	series = {Cambridge Studies in Modern Optics},
	volume = {22},
	publisher = {Cambridge University Press},
	address = {Cambridge},
	year = {1997}
}

@book{Paris.Rehacek:Book:2004,
	author = {M. Paris and J. Řeháček (Eds.)},
	title = {Quantum State Estimation},
    series = {Lecture Notes in Physics},
	volume = {649},
	publisher = {Springer Berlin, Heidelberg},
    doi = {10.1007/b98673},
	year = {2004}
}

@book{Conway:Book:2007,
	author = {J. B. Conway},
	title = {A Course in Functional Analysis},
	series = {Graduate Texts in Mathematics},
	volume = {96},
	edition = {2},
	publisher = {Springer},
	address = {New York},
	year = {2007},
	doi = {10.1007/978-1-4757-4383-8}
}

@book{Higham:Book:2008,
	author = {N. J. Higham},
	title = {Functions of Matrices: Theory and Computation},
	publisher = {Society for Industrial and Applied Mathematics},
	address = {Philadelphia, PA},
	year = {2008},
	doi = {10.1137/1.9780898717778}
}

@book{Nielsen.Chuang:Book:2010,
	author = {M. A. Nielsen and I. L. Chuang},
	title = {Quantum Computation and Quantum Information},
	publisher = {Cambridge University Press},
	doi = {10.1017/cbo9780511976667},
	year = {2010}
}

@book{Breuer.Petruccione:book:2002,
  author    = {Breuer, Heinz-Peter and Petruccione, Francesco},
  title     = {The Theory of Open Quantum Systems},
  publisher = {Oxford University Press},
  year      = {2002}
}

	\clearpage
	\section*{Supplementary Material} \label{app:supplementary.material}
	
	\setcounter{secnumdepth}{0}
	
	\titleformat{\subsection}[hang]
	{\normalfont\normalsize\bfseries}
	{}{0pt}{\raggedright}
	
	\appendixtoc
	\appendixtocentriestrue
	
	\subsection{Modular-frequency and integral forms of the inverse KM action}
	\label{app:subsec:KM.modular.form}
	Let $Y \!=\! \sum_i y_i\ketbra{i}{i}$, $y_i \!>\! 0$. Then each matrix unit $\ketbra{i}{j}$ is an eigenoperator of the modular generator $K_Y \defeq [\ln Y,\cdot]$ 
	with eigenvalue
	\begin{equation}
		K_Y(\ketbra{i}{j}) = (\ln y_i-\ln y_j)\ketbra{i}{j} \eqdef \omega_{ij}\ketbra{i}{j}.
	\end{equation}
	In the same basis, the inverse KM action is diagonal,
	\begin{equation}
		\big(\Omega_Y^{-1}\!(X)\big)_{ij} = \frac{\ln y_i-\ln y_j}{y_i-y_j}\,X_{ij},
	\end{equation}
	with the diagonal case understood by continuity. Since
	\begin{equation}
		y_i - y_j = 2\, \sqrt{y_i y_j}\, \sinh(\omega_{ij}/2),
	\end{equation}
	we obtain
	\begin{equation}
		\big(\Omega_Y^{-1}\!(X)\big)_{ij} = k(\omega_{ij})\,\frac{X_{ij}}{\sqrt{y_i y_j}},
		\quad
		k(\omega) = \frac{\omega}{2\sinh(\omega/2)}.
	\end{equation}
	Equivalently,
	\begin{equation}
		\Omega_Y^{-1}\!(X) = k(K_Y)\!\big[Y^{-1/2}XY^{-1/2}\big].
	\end{equation}
	This proves the modular-frequency filter form in Theorem~\ref{thm:KM.update}.
	
	The same inverse KM action admits the integral representation
	\begin{equation}
		\Omega_Y^{-1}\!(X) = \int_{\mathbb{R}} \beta_0(t)\, Y^{-(1+it)/2}X\,Y^{-(1-it)/2}\,dt,
		\label{eq:KM.integral.representation}
	\end{equation}
	where $\beta_0(t) = \pi/[2(\cosh(\pi t)+1)]$
	is normalized such that
	\begin{equation}
		k(\omega) = \int_{\mathbb{R}} \beta_0(t)e^{-it\omega/2}\,dt.
	\end{equation}
	The same probability density $\beta_0$ appears in Ref.~\cite{Junge.etal:AnnHenriPoincare:2018} as the weight in the averaged rotated-Petz recovery map. Here, however, it represents the inverse KM operator itself rather than a CPTP average of recovery channels.
	
	\subsection{Properties of the KM update} 
	\label{app:subsec:KM.update.properties}
	\begin{enumerate}[wide, labelwidth=!, labelindent=0pt, label=\textbf{\arabic*.}]
		\item \textbf{Linearity in $\bm{X}$}
		\begin{equation}
			\mathcal{R}^{\text{KM}}_{\sigma,\mathcal{E}}(aX_1+bX_2) 
			= a\,\mathcal{R}^{\text{KM}}_{\sigma,\mathcal{E}}(X_1) + b\,\mathcal{R}^{\text{KM}}_{\sigma,\mathcal{E}}(X_2).
		\end{equation}
		\item \textbf{Positivity} \mbox{} \\ 
		If $X \geqslant 0$, then
		\begin{equation}
			\mathcal{R}^{\text{KM}}_{\sigma,\mathcal{E}}(X) \geqslant 0,
		\end{equation}
		because $\Omega^{-1}_Y$ is positive, $\mathcal{E}^{\dagger}$ is positive, and conjugation by $\sqrt{\S}$ preserves positivity. 
		\item \textbf{Trace preservation in the argument}
		\begin{align}
			\Tr\!\big[\mathcal{R}^{\text{KM}}_{\sigma,\mathcal{E}}(X)\big]
			&= \Tr\!\big[\sigma\, \mathcal{E}^{\dagger}[\Omega_Y^{-1}\!(X)]\big] =\Tr\!\big[\mathcal{E}(\sigma)\, \Omega_Y^{-1}\!(X)\big] \nonumber \\
			&= \Tr\!\big[Y\, \Omega^{-1}_Y\!(X)\big] = \Tr[X].
			\label{eq:KM.trace.preservation}
		\end{align}
		In particular, if $X$ is a density operator, then $\mathcal{R}^{\text{KM}}_{\S,\mathcal{E}}(X)$ is again a density operator.
		\item \textbf{Exact recovery of the reference output}
		\begin{align}
			\mathcal{R}^{\text{KM}}_{\S,\mathcal{E}}\big(\mathcal{E}(\S)\big) = \S,
		\end{align}
		since $\Omega^{-1}_Y(Y)=\mathds{1}$ and $\mathcal{E}^{\dagger}(\mathds{1})=\mathds{1}$.
		\item \textbf{Reduction to the Petz term in the commuting limit} \mbox{} \\
		If $[X,Y]=0$, then
		\begin{align}
			\Omega^{-1}_Y\!(X)=Y^{-1/2} X Y^{-1/2},
		\end{align}
		and therefore
		\begin{align}
			\mathcal{R}^{\text{KM}}_{\S,\mathcal{E}}(X) = \sqrt{\sigma}\; \mathcal{E}^{\dagger}\!\big[ Y^{-1/2} X Y^{-1/2} \big]\, \sqrt{\sigma},
		\end{align}
		which is the ordinary Petz recovery form.
		
		\item \textbf{Sum representation} \mbox{} \\
		The KM update can be written as the sum
		\begin{align}
			\mathcal{R}^{\text{KM}}_{\S,\mathcal{E}}(X) = \S + \Delta^{\text{KM}}_{\S,\mathcal{E}}(X),
		\end{align}
		with $\Delta^{\text{KM}}_{\S,\mathcal{E}}(X) \equiv G_{\!X}^\sigma(\sigma) - \sigma$, cf.\ Eqs.~\eqref{eq:delta.gradient} and \eqref{def:KM.update.endmatter}.
		
		\item \textbf{Channel representation} \mbox{} \\
		Introducing
		\begin{align}
			M_\sigma(A) &\defeq \mathcal{E} \big[\sqrt{\sigma}\,\mathcal{E}^\dagger(A)\,\sqrt{\sigma}\big], \\
			\Phi_Y(A) &\defeq Y^{-1/2}\, M_\sigma(A)\; Y^{-1/2},
		\end{align}
		the predicted output after the update $Y_+ \defeq \mathcal{E}[\mathcal{R}^{\text{KM}}_{\sigma,\mathcal{E}}(X)]$ factorizes as
		\begin{equation}
			Y_+ = M_\S[\Omega^{-1}_Y\!(X)] = \Phi_{Y,*}(\widehat{X}_Y),
		\end{equation}
		where $\Phi_{Y,*}$ denotes the Hilbert--Schmidt adjoint of $\Phi_{Y}$ and $\widehat{X}_Y \defeq \sqrt{Y}\,\Omega^{-1}_Y\!(X) \sqrt{Y}$. Since $\Phi_Y$ is unital CP, $\Phi_{Y,*}$ is CPTP. Moreover, it preserves $Y$, $\Phi_{Y,*}(Y)=Y$. This is not the physical measurement channel itself, but an auxiliary CPTP map determined by $Y$, $\S$, and $\mathcal{E}$. The data processing inequality can be written in CPTP channel form as
		\begin{equation}
			D (Y_+ \Vert Y) = D \big(\Phi_{Y,*}\big(\widehat{X}_Y\big) \big\Vert \Phi_{Y,*}(Y) \big) \leqslant D\big(\widehat{X}_Y \Vert Y\big).
		\end{equation}
	\end{enumerate}
	
	\subsection{Concavity of the log-likelihood and uniqueness of the maximizer} \label{app:subsec:log.likelihood.concavity}
	For fixed $X\geqslant0$, the map
	\begin{equation}
		\sigma\mapsto \mathcal L_X(\sigma) = \Tr[X \ln \mathcal{E}(\sigma)]
	\end{equation}
	is concave on the domain where $\mathcal{E}(\sigma)>0$. Indeed, for $0\leqslant t\leqslant 1$,
	\begin{align}
		\mathcal{L}_X(t\sigma_1+(1-t)\sigma_0) 
		&= \Tr\!\big[ X \ln\!\big(t \mathcal{E}(\sigma_1) + (1-t) \mathcal {E}(\sigma_0)\big) \big] \nonumber\\
		& \geqslant t \mathcal{L}_X(\sigma_1) + (1-t) \mathcal{L}_X(\sigma_0),
	\end{align}
	by operator concavity of the logarithm and positivity of the trace pairing with $X$. Consequently, for any two admissible states,
	\begin{equation}
		\mathcal {L}_X(\rho) \leqslant \mathcal{L}_X(\sigma) + D\mathcal{L}_X(\sigma)[\rho-\sigma].
		\label{eq:derivative_inequality}
	\end{equation}
	
	In the quantum-to-classical case, write $p_i \defeq \hat{p}_i$ and set
	\begin{align}
		\mathcal{S} &\defeq \{\sigma\geqslant0:\Tr[\sigma]=1\}, \\
		\mathcal{D}_+ &\defeq \{\sigma\in\mathcal{S}:\Tr[\Pi_i\sigma]>0\ \;\forall\, i\}.
	\end{align}
	On states in $\mathcal{S} \!\setminus\! \mathcal{D}_+$, the likelihood $\mathcal{L}_{\text{qc}}$ is extended by setting $\mathcal{L}_{\text{qc}}(\sigma)=-\infty$.
	
	\begin{lemma} \label{lem:qc.unique.MLE}
		Suppose $\{\Pi_i\}$ is a finite informationally complete POVM with nonzero effects and $p_i \!>\! 0 \;\forall\, i$. Then $\mathcal{L}_{\text{qc}}$ is strongly concave on $\mathcal{D}_+$ and has a unique maximizer over the state space $\mathcal{S}$.
	\end{lemma}
	\begin{proof}
		Let $\mathbb{H}_0$ denote the space of traceless Hermitian matrices and set
		\begin{equation}
			p_{\min} \defeq \min_i p_i, \quad
			c \defeq \!\!\min\limits_{\substack{\Delta \in \mathbb{H}_0\\ \|\Delta\|_2=1}} \sum_i \;[\Tr(\Pi_i\Delta)]^2 .
		\end{equation}
		Since $\{\Pi_i\}$ is informationally complete, the map $\Delta \!\mapsto\! (\Tr[\Pi_i\Delta])_i$ is injective on $\mathbb{H}_0$. Hence $c \!>\! 0$, because the function being minimized is continuous and strictly positive on the compact unit sphere in $\mathbb{H}_0$. For $\sigma\in\mathcal{D}_+$ and $\Delta\in\mathbb{H}_0$,
		\begin{align}
			\begin{aligned}
				-D^2\mathcal{L}_{\text{qc}}(\sigma)[\Delta,\Delta] &= \sum_i \frac{p_i}{[\Tr(\Pi_i\sigma)]^2} [\Tr(\Pi_i\Delta)]^2 \\
				& \geqslant \sum_i \,p_i[\Tr(\Pi_i\Delta)]^2 \\
				& \geqslant p_{\min}\, c\, \|\Delta\|_2^2,
			\end{aligned}
		\end{align}
		where we used $\Tr[\Pi_i\sigma] \leqslant 1$ due to $0 \leqslant \Pi_i \leqslant \mathds{1}$ and $\Tr[\sigma]=1$. Thus $\mathcal{L}_{\text{qc}}$ is strongly concave on $\mathcal{D}_+$.
		
		It remains to prove existence. Since each $\Pi_i$ is a nonzero POVM effect, the maximally mixed state $\mathds{1}/d$ belongs to $\mathcal{D}_+$, and thus $\mathcal{L}_{\text{qc}}(\mathds{1}/d) \!>\! -\infty$. With the boundary convention specified above, the map $\sigma \!\mapsto\! \mathcal{L}_{\text{qc}}(\sigma)$ is upper semicontinuous on $\mathcal{S}$, i.e.,
		\begin{equation}
			\limsup_{n\to\infty} \;\mathcal{L}_{\text{qc}}(\sigma_n) \leqslant \mathcal{L}_{\text{qc}}(\sigma)
		\end{equation}
		whenever $\sigma_n \!\in\! \mathcal{S}$ and $\sigma_n \!\to\! \sigma \!\in\! \mathcal{S}$. Since $\mathcal{S}$ is compact, $\mathcal{L}_{\text{qc}}$ attains a maximum on $\mathcal{S}$. Moreover, since the maximum value is at least $\mathcal{L}_{\text{qc}}(\mathds{1}/d) \!>\! -\infty$, any maximizer must lie in $\mathcal{D}_+$.
		
		Finally, if $\sigma_1$ and $\sigma_2$ were two distinct maximizers, then for $0<\theta<1$ strong concavity would give
		\begin{equation}
			\mathcal{L}_{\text{qc}} \big((1-\theta) \sigma_1 + \theta \sigma_2 \big)
			>
			(1-\theta) \mathcal{L}_{\text{qc}}(\sigma_1) + \theta \mathcal{L}_{\text{qc}}(\sigma_2),
		\end{equation}
		contradicting maximality. Hence the maximizer is unique.
	\end{proof}
	
	\subsection{Proof of Theorem~\ref{minimum_increase}}
	\paragraph*{Lower bound.}
	Assume that $Y>0$ and $Y_\eta>0$, such that the logarithms below are ordinary matrix logarithms. We have
	\begin{equation}
		\Tr[X \ln Y_\eta - X \ln Y]=D(X||Y)-D(X||Y_\eta),
	\end{equation}
	where $Y=\mathcal{E}(\sigma)$ and $Y_\eta=\mathcal{E}(\sigma_\eta)$. The quantum relative entropy, defined as~\cite{Umegaki:KodaiMathSemRep:1962,Nielsen.Chuang:Book:2010},
	\begin{equation}\label{eq:quantum_relative_entropy}
		D(X||Y)=\Tr[X(\ln X-\ln Y)],
	\end{equation}
	can be expressed through the following optimization identity~\cite{Ruskai:JMathPhys:2002}:
	\begin{equation}
		D(X||Y) = \sup_{H=H^\dagger} \big\lbrace \! \Tr[XH] - \ln \Tr\big[e^{H+\ln Y}\big] \big\rbrace.
	\end{equation}
	This implies
	\begin{align}
		\Tr[X \ln Y_\eta - X \ln Y] &= \sup_{H=H^\dagger} \big\lbrace \! \Tr[XH] - \ln \Tr[e^{H+ \ln Y}] \big\rbrace \nonumber \\
		& \hspace*{-5mm} - \sup_{\tilde{H} = \tilde{H}^\dagger} \big\lbrace \! \Tr\!\big[X \tilde{H}\big] - \ln \Tr\!\big[e^{\tilde{H} + \ln Y_\eta}\big] \big\rbrace \nonumber \\
		& \hspace*{-12mm} \geqslant -\ln \Tr[e^{\ln X - \ln Y_\eta + \ln Y}] + \ln \Tr[e^{\ln X}] \nonumber \\
		& \hspace*{-12mm} = - \ln \Tr[e^{\ln X - \ln Y_\eta + \ln Y}],
	\end{align}
	where we picked $H = \tilde{H} = \ln X-\ln Y_\eta$, which maximizes the second supremum and gives a lower bound on the first.
	
	We use Lieb's triple-matrix inequality~\cite{Lieb:AdvMath:1973,Ruskai:JMathPhys:2002}: for any positive definite operators $R,S,T>0$,
	\begin{equation}
		\Tr\!\left[e^{\ln R-\ln S+\ln T}\right]  \leqslant \Tr\!\left[ R\,\Omega_S^{-1}\!(T) \right],
	\end{equation}
	where
	\begin{equation}
		\Omega_S^{-1}\!(T) = \int_0^\infty (S+u\mathds{1})^{-1}\,T\,(S+u\mathds{1})^{-1} \,du .
	\end{equation}
	Using the theorem, we thus have
	\begin{equation}
		\Tr[X \ln Y_\eta - X \ln Y] \geqslant - \ln \Tr[ X\, \Omega_{Y_\eta}^{-1}\!(Y)].
	\end{equation}
	Next, we derive,
	\begin{align}
		& 1 - \Tr[X\, \Omega_{Y_\eta}^{-1}\!(Y)] \nonumber \\
		& = \Tr[\mathcal{E}^\dagger \big(\Omega_{Y_\eta}^{-1}\!(X)\big)(\sigma_\eta-\sigma)] \\
		& = \eta \Tr[\big[\mathcal{E}^\dagger \big(\Omega_{Y_\eta}^{-1}\!(X)\big) - \mathds{1}\big](\sigma_+ - \sigma)] \nonumber \\
		& = \eta \Tr[\big[ \mathcal{E}^\dagger\!\big(\Omega_{Y_\eta}^{-1}\!(X)\big)-\mathds{1}\big] \big(\sqrt{\sigma}\, \mathcal{E}^{\dagger} \!\big(\Omega_{Y}^{-1}\!(X)\big) \sqrt{\sigma} - \sigma \big)] \nonumber \\
		&=\eta \textsl{g}_{\sigma} \Big( \sqrt{\sigma}\big[ \mathcal{E}^\dagger\!\big(\Omega_{Y_\eta}^{-1}\!(X)\big) - \mathds{1} \big] \sqrt{\sigma} , \sqrt{\sigma} \big[ \mathcal{E}^\dagger\!\big(\Omega_{Y}^{-1}\!(X)\big)-\mathds{1} \big] \sqrt{\sigma} \Big) \nonumber \\
		&= \eta \textsl{g}_{\sigma}(\Delta_{\sigma_\eta}^\sigma,\Delta_{\sigma}^\sigma). \nonumber
	\end{align}
	This provides the desired inequality,
	\begin{equation}
		\hspace*{-0.75mm} \Tr[X \ln Y_\eta - X \ln Y]
		\geqslant
		-\ln\!\big[1-\eta\,\textsl{g}_{\sigma}(\Delta_{\sigma_\eta}^\sigma,\Delta)\big],
	\end{equation}
	where we have set $\Delta \equiv \Delta_{\sigma}^\sigma$.
	
	Using the inequality $-\ln(1-x)\geqslant x$, valid whenever $x<1$, gives the looser inequality
	\begin{equation}
		\Tr[X \ln Y_\eta \!-\! X \ln Y] \geqslant 1 \!-\! \Tr\big[X\, \Omega_{Y_\eta}^{-1}\!(Y)\big]\!=\!\eta \textsl{g}_{\sigma}(\Delta_{\sigma_\eta}^\sigma,\Delta).
	\end{equation}
	This can be independently derived using concavity of the log-likelihood \eqref{eq:derivative_inequality}, and extends the inequality to semidefinite $X$ for which the explicit optimizer $H=\ln X-\ln Y_\eta$ has to be interpreted by approximation.
	
	\paragraph*{Upper bound.}
	Using Eq.~\eqref{eq:derivative_inequality}, we have
	\begin{align}
		\mathcal{L}_X(\sigma_\eta) - \mathcal{L}_X(\sigma) & \leqslant D\mathcal{L}_X(\sigma)[\sigma_\eta-\sigma] \\
		&= \eta\,\Tr\!\left[\mathcal{E}^{\dagger}\!\big(\Omega_Y^{-1}\!(X)\big)\,\Delta\right] \nonumber \\
		&= \eta\,\Tr\!\left[\left(\mathcal{E}^{\dagger}\!\big(\Omega_Y^{-1}\!(X)\big) - \mathds{1}\right) \Delta\right] \nonumber \\
		&= \eta\,\textsl{g}_{\sigma}(\Delta,\Delta), \nonumber
	\end{align}
	where $Y=\mathcal{E}(\sigma)$ and $\Delta=\sigma_+-\sigma=(\sigma_\eta-\sigma)/\eta$.
	
	\subsection{Evidence for Conjecture 3}
	We tested the conjectured monotonicity of the KM update by evaluating
	\begin{equation}
		\Delta_{\mathrm{KM}}
		\equiv
		\mathcal{L}_X(\sigma_+)-\mathcal{L}_X(\sigma)
	\end{equation}
	for randomly generated states and quantum channels.
	
	For the broad random search, the states $X$ and $\sigma$ were sampled independently from a regularized Hilbert--Schmidt ensemble,
	\begin{equation}
		\rho = \frac{AA^\dagger+\epsilon I}{\Tr[AA^\dagger+\epsilon I]},
		\qquad
		\epsilon=10^{-8},
	\end{equation}
	where $A$ is a complex Ginibre matrix. The channels were generated from
	independent complex Ginibre Kraus operators $L_a$, normalized according to
	\begin{equation}
		K_a
		=
		L_a
		\left(
		\sum_b L_b^\dagger L_b
		\right)^{-1/2},
	\end{equation}
	which guarantees $\sum_a K_a^\dagger K_a=I$. We tested dimensions $d=2,\ldots,5$ and Kraus ranks $2,\ldots,5$, for a total of $10^8$ independent random instances. No violation of $\Delta_{\mathrm{KM}}\geq0$ was observed. The smallest values occurred predominantly for Kraus-rank-two channels.
	
	Because uniform random sampling may fail to detect a small region in parameter space where $\Delta_{\mathrm{KM}}<0$, we supplemented this test by a targeted optimization. Since the smallest margins in the random search occurred for Kraus-rank-two channels, we concentrated this search on $d=2$ and Kraus rank two. We first generated $10^4$ independent random triples $(\sigma,X,\mathcal E)$ from the same regularized ensemble and selected the one with the smallest $\Delta_{\mathrm{KM}}$ as the initial point. In this initial sample the minimum was
	\begin{equation}
		\Delta_{\mathrm{KM}}
		\simeq 1.05\times10^{-4}.
	\end{equation}
	
	Starting from this candidate, we performed a multistart simulated-annealing search directly minimizing $\Delta_{\mathrm{KM}}$. A trial move consisted of a stochastic perturbation of either $\sigma$, $X$, the Kraus operators, or two of these objects simultaneously, with respective probabilities $0.35$, $0.35$, $0.20$, and $0.10$. After every perturbation, the states were renormalized to density matrices and the Kraus operators were renormalized to satisfy the trace-preserving condition. During this stage the density-matrix parametrization was regularized by adding $10^{-10} I$ before normalization.
	
	The perturbation amplitudes for $\sigma$, $X$, and the Kraus operators were all decreased geometrically from $5\times10^{-2}$ to $5\times10^{-3}$ over the course of each run. To combine local exploration with occasional larger moves, the perturbation amplitude was additionally multiplied by $0.3$, $1$, or $3$, chosen with probabilities $0.70$, $0.25$, and $0.05$, respectively.
	
	Moves that decreased $\Delta_{\mathrm{KM}}$ were always accepted. Moves that increased it by $\delta$ were accepted with the simulated-annealing probability
	\begin{equation}
		p_{\mathrm{acc}} = \exp\!\left(-\frac{\delta}{T}\right),
	\end{equation}
	where the effective temperature $T$ was decreased geometrically from $10^{-2}$ to $10^{-5}$. This allows a trajectory to leave a local minimum during the early part of the search. If no improvement of the best value was found for $5000$ consecutive steps, the current point was restarted by a larger perturbation around the best point found so far.
	
	We performed $40$ independent simulated-annealing runs, each containing $2\times10^5$ trial moves and using an independent random seed. The smallest value obtained over all $40$ runs was
	\begin{equation}
		\Delta_{\mathrm{KM}}
		\simeq 3.10\times10^{-10}.
	\end{equation}
	
	Finally, we subjected the best annealing result to a zero-temperature local refinement. In this stage, only moves decreasing $\Delta_{\mathrm{KM}}$ were accepted. The perturbation scales were decreased geometrically from $5\times10^{-3}$ to $5\times10^{-4}$, and the refinement was run for $10^5$ steps. The relative probabilities of perturbing $\sigma$, $X$, the channel, or two quantities simultaneously were $0.40$, $0.40$, $0.15$, and $0.05$, respectively, with a restart after $10^4$ steps without improvement.
	
	The local refinement reached
	\begin{equation}
		\Delta_{\mathrm{KM}}
		\simeq -2.22\times10^{-16}.
	\end{equation}
	This value is at the scale of double-precision machine accuracy and is therefore numerically indistinguishable from zero. The corresponding state was close to the boundary of state space, with smallest eigenvalue
	approximately $1.9\times10^{-10}$, and was nearly a fixed point of the KM iteration, $\sigma_+\simeq\sigma$. We therefore interpret the residual negative value as floating-point roundoff rather than evidence of a violation.
	
	For numerical stability, eigenvalues entering logarithms and inverse matrix functions were clipped below at $10^{-15}$. The logarithmic mean appearing in $\Omega^{-1}$ was evaluated with a relative near-degeneracy criterion of $100\,\epsilon_{\mathrm{mach}}$, where $\epsilon_{\mathrm{mach}}\simeq2.22\times10^{-16}$. These numerical safeguards are distinct from the $10^{-8}$ and $10^{-10}$ state regularizations used in the random and annealing ensembles, respectively.
	
	Thus, neither the broad search over $10^8$ random instances nor the targeted multistart optimization produced a numerically significant value $\Delta_{\mathrm{KM}}<0$.
	
	\subsection{Proof of Theorem~\ref{thm:KM.fixed.point.MLE}}
	We will prove that every full-rank fixed point $\sigma$ of the KM update is a maximizer of the log-likelihood. Since $\sigma$ is full rank, the fixed-point equation $\mathcal{R}^{\text{KM}}_{\sigma,\mathcal{E}}(X)=\sigma$ is equivalent to $\GX(\sigma)=\mathds{1}$. Conversely, stationarity on the trace-one manifold gives $\GX(\sigma)=\lambda\mathds{1}$, and taking the trace against $\sigma$ gives $\lambda=\Tr[X]=1$. Thus being a full-rank fixed point is equivalent to being a stationary point of the log-likelihood. In other words, as we have shown in the End Matter, the following statements are equivalent,
	\begin{align}
		\begin{aligned}
			D\mathcal{L}_X(\sigma)[\Delta']=0 
			\quad 
			\forall\; \Delta'=\Delta'^\dagger,\; \Tr[\Delta']=0\\
			\Longleftrightarrow\, \GX(\sigma)=\mathds{1} \,\Longleftrightarrow\, \Delta=0 \,\Longleftrightarrow\, \mathcal{R}^{\text{KM}}_{\sigma,\mathcal{E}}(X)=\sigma.
		\end{aligned}
	\end{align}
	This follows from Eqs.~\eqref{eq:delta.gradient} and \eqref{def:KM.update.endmatter}. Using concavity of the log-likelihood \eqref{eq:derivative_inequality}, for any other state $\sigma'$, we have
	\begin{equation}
		\mathcal{L}_X(\sigma') \leqslant \mathcal{L}_X(\sigma) + D\mathcal{L}_X(\sigma)[\sigma'-\sigma] = \mathcal{L}_X(\sigma),
	\end{equation}
	where we selected $\Delta'=\sigma'-\sigma$ for the first stationarity condition. Thus, the log-likelihood of any other state is upper-bounded by that at the stationary point; hence every full-rank stationary point is a maximizer of the log-likelihood.
	
	\subsection{Proof of Theorems~\ref{thm:Petzqc} and~\ref{thm:Petz.converg}}
	\label{app:subsec:proofs.Petzqc.Petz.iteration.convergence}
	In what follows, we use $p_i \!\defeq\! \hat{p}_i \!\geqslant\! 0$ with $\sum_i p_i \!=\! 1$, $q_i \!\defeq\! \Tr[\Pi_i\sigma]>0$, $q_i'\defeq\Tr[\Pi_i\sigma_+]$, and $\| \!\cdot\!\|$ to denote the Hilbert--Schmidt norm. Terms with $p_i=0$ are understood as omitted in expressions containing $(Cp)_i$ or $q_i'$ in a denominator, or containing $q_i'$ inside a logarithm or inside $f(q_i'/q_i)$. Let
	\begin{equation}
		p = \begin{pmatrix}
			p_1 \\
			\vdots \\
			p_m
		\end{pmatrix},\quad
		q = \begin{pmatrix}
			q_1 \\
			\vdots \\
			q_m
		\end{pmatrix},\quad
		Q = \begin{pmatrix}
			q_1 & & 0 \\
			& \ddots & \\
			0 & & q_m
		\end{pmatrix}.
	\end{equation}
	Next, define the $m \!\times\! m$ positive semidefinite matrix $B$ by
	\begin{equation}
		B_{ij} \defeq \frac{1}{q_i q_j}\Tr[\Pi_i\sqrt{\sigma}\, \Pi_j\sqrt{\sigma}].
		\label{def:Gram.matrix}
	\end{equation}
	The above is a Gram matrix: its entries are given by $B_{ij} = \langle \frac{1}{q_i}\Pi_i,\frac{1}{q_j}\Pi_j\rangle_\sigma$, where
	\begin{equation}
		\langle A,B\rangle_\sigma \defeq \Tr[A^\dagger \sqrt{\sigma}\, B \sqrt{\sigma}]
	\end{equation}
	is a positive semidefinite sesquilinear form on $\mathbb{C}^{d\times d}$. When $\sigma$ is full rank, it is an inner product, known as the KMS inner product~\cite{Junge:JStatPhys:2023,Vernooij.Wirth:CommunMathPhys:2023}. We also let $C=QB$. Thus, the entries of $C$ are given by
	\begin{equation}
		C_{ij} = \frac{1}{q_j} \Tr[\Pi_i\sqrt{\sigma}\, \Pi_j\sqrt{\sigma}].
	\end{equation}
	These entries are nonnegative, since
	\begin{equation}
		\Tr[\Pi_i \sqrt{\sigma}\,\Pi_j\sqrt{\sigma}] = \left\|\Pi_i^{1/2}\sqrt{\sigma}\; \Pi_j^{1/2}\right\|^2 \geqslant 0 .
	\end{equation}
	Note that $Cq=q$, since
	\begin{align}
		\begin{aligned}
			(Cq)_i & =\sum_j C_{ij}\, q_j = \sum_j \frac{1}{q_j}\Tr[\Pi_i\sqrt{\sigma}\, \Pi_j\sqrt{\sigma}]q_j \\
			&=\Tr\Big[\Pi_i\sqrt{\sigma}\, \sum_j\Pi_j\sqrt{\sigma}\Big] =\Tr[\Pi_i\sigma] = q_i.
		\end{aligned}
	\end{align}
	Finally, the spectral radius $\rho(A)$ of a square matrix $A$ is defined as
	\begin{equation}
		\rho(A) = \max\{|\lambda|:\text{$\lambda$ is an eigenvalue of $A$}\}.
	\end{equation}
	
	\begin{lemma}\label{lem:points}
		If $D$ is a diagonal matrix with diagonal entries $d_1,\dots,d_m \geqslant 0$, then $\sum_i d_i q_i \leqslant \rho(DC)$.
	\end{lemma}
	\begin{proof}
		The matrix $H\defeq \sqrt{B}DQ\sqrt{B}$ is Hermitian positive semidefinite. Hence, for every nonzero vector $z$,
		\begin{equation}
			\frac{\langle Hz,z\rangle}{\langle z,z\rangle}\leqslant \lambda_{\max}(H) = \rho(H).
		\end{equation}
		Moreover, $H$ and $DQB \!=\! DC \!=\! (DQ\sqrt{B})\sqrt{B}$ have the same nonzero eigenvalues, and therefore the same spectral radius. Thus $\rho(H)=\rho(DC)$. Setting $z=\sqrt{B}q$ and using $Cq=q$ and $Bq=\mathbf{1}$, where $\mathbf{1}$ is the vector with every entry equal to one, gives
		\begin{align}
			\begin{aligned}
				\rho(DC) &\geqslant \frac{\langle \sqrt{B}DCq,\sqrt{B}q\rangle}{\langle \sqrt{B}q,\sqrt{B}q\rangle}
				= \frac{\langle Dq,Bq\rangle}{\langle q,Bq\rangle} \\
				&= \frac{\langle Dq,\mathbf{1}\rangle}{\langle q,\mathbf{1}\rangle}
				= \frac{\sum_i d_iq_i}{\sum_i q_i}
				= \sum_i d_iq_i .
			\end{aligned}
		\end{align}
	\end{proof}
	
	\begin{lemma}\label{lem:specrad}
		Let $A$ be a square matrix such that $A_{ij} \geqslant 0$, and let $x$ be a vector such that $x_i>0$. If $Ax \leqslant \lambda x$ with $\lambda>0$, then $\rho(A) \leqslant \lambda$. In particular, if $Ax=\lambda x$, then $\rho(A)=\lambda$.
	\end{lemma}
	\begin{proof}
		See Theorem 4 in Ref.~\cite{Jacobs:ElectronProcTheorComputSci.2021} for a proof of the first statement. If $Ax=\lambda x$, then $\lambda$ is an eigenvalue of $A$, so $\rho(A) \geqslant \lambda$ and $\rho(A) \leqslant \lambda$. Hence, $\rho(A)=\lambda$.
	\end{proof}
	
	\begin{lemma}[Jacobs-type inequality for the quantum-to-classical Petz update]\label{lem:Jacobs}
		$\Tr[\GX(\sigma_+)\, \sigma] \leqslant 1$.
	\end{lemma}
	\begin{proof}
		The proof adapts the matrix argument behind Jacobs' Proposition~2 on the classical Jeffrey-rule contraction~\cite{Jacobs:ElectronProcTheorComputSci.2021} to the quantum-to-classical Petz update. The new ingredient is the Gram matrix \eqref{def:Gram.matrix} built from the noncommuting POVM elements and the current state $\sigma$.
		
		Let $D$ be the diagonal matrix defined by the entries
		\begin{equation}
			d_i \defeq
			\begin{cases}
				p_i/[(Cp)_i], & p_i>0,\\
				0, & p_i=0.
			\end{cases}
		\end{equation}
		For $p_i>0$, the denominator is positive because
		\begin{equation}
			(Cp)_i \geqslant C_{ii}p_i = \frac{p_i}{q_i} \left\|\Pi_i^{1/2}\sqrt{\sigma}\,\Pi_i^{1/2}\right\|^2 >0 ,
		\end{equation}
		and we have
		\begin{equation}
			\sum_{j:p_j>0} d_i C_{ij}p_j = d_i\sum_j C_{ij}p_j = d_i(Cp)_i = p_i .
		\end{equation}
		Equivalently, the submatrix of $DC$ obtained by retaining the rows and columns with $p_i \!>\! 0$ satisfies
		\begin{equation}
			(DC)_{p_i>0,p_j>0}\,(p_j)_{p_j>0}=(p_i)_{p_i>0}.
		\end{equation}
		Thus this submatrix is entrywise nonnegative and has the strictly positive eigenvector $(p_i)_{p_i>0}$ with eigenvalue one. By Lemma~\ref{lem:specrad}, its spectral radius is one. It remains only to compare this retained submatrix with the full matrix $DC$. If the indices with $p_i \!>\! 0$ are ordered first, then
		\begin{equation}
			DC =
			\begin{pmatrix}
				(DC)_{p_i>0,p_j>0} & (DC)_{p_i>0,p_j=0} \\
				0 & 0
			\end{pmatrix},
		\end{equation}
		because $d_i \!=\! 0$ whenever $p_i \!=\! 0$. Hence the full matrix has the same nonzero eigenvalues as the retained submatrix, together with possible additional zero eigenvalues. Therefore $\rho(DC) \!=\! 1$.
		
		Thus, by Lemma~\ref{lem:points},
		\begin{align}
			\Tr[\GX(\sigma_+)\, \sigma]
			&= \sum_{i:p_i>0} \frac{p_iq_i}{\Tr[\Pi_i \sigma_+]} \nonumber \\
			&= \sum_{i:p_i>0} \frac{p_iq_i}{\sum_j \frac{p_j}{q_j} \Tr[\Pi_i\sqrt{\sigma}\, \Pi_j\sqrt{\sigma}]} \nonumber \\
			&= \sum_{i:p_i>0} \frac{p_i q_i}{(Cp)_i}
			= \sum_i d_i q_i \\
			&\leqslant \rho(DC) = 1. \nonumber
		\end{align}
	\end{proof}
	
	\begin{lemma}\label{lem:control}
		For every $\epsilon>0$, there exists $\delta>0$ such that $\Tr[\GX(\sigma)\, \sigma_+]-1<\epsilon$ whenever $\mathcal{L}_\text{qc}(\sigma_+)-\mathcal{L}_\text{qc}(\sigma)<\delta$.
	\end{lemma}
	\begin{proof}
		Set $q_i=\Tr[\Pi_i\sigma]$ and $q_i^{\prime}=\Tr[\Pi_i\sigma_+]$. By Lemma~\ref{lem:Jacobs}, we have
		\begin{equation}\label{eq:bw}
			\sum\limits_i p_i \frac{q_i}{q_i^{\prime}} = \Tr[\GX(\sigma_+)\, \sigma] \leqslant 1.
		\end{equation}
		Define $f \!:\! (0,\infty)\to\mathbb{R}$ by $f(t)=t^{-1}-1+\ln t$. Since
		\begin{equation}
			f'(t)=\frac{t-1}{t^2}
		\end{equation}
		and $f(1)=0$, the function $f$ is nonnegative on $(0,\infty)$ and increasing on $[1,\infty)$. By \eqref{eq:bw},
		\begin{align}
			\sum\limits_i p_i f\!\left(\frac{q_i^{\prime}}{q_i}\right) &=
			\sum\limits_i p_i\frac{q_i}{q_i^{\prime}} - 1 + \sum\limits_i p_i \ln\!\left(\frac{q_i^{\prime}}{q_i}\right) \\
			&\leqslant \sum\limits_i p_i\ln\!\left(\frac{q_i^{\prime}}{q_i}\right) =
			\mathcal{L}_\text{qc}(\sigma_+)-\mathcal{L}_\text{qc}(\sigma). \nonumber
		\end{align}
		Given $\epsilon \!>\! 0$, choose $\delta \!=\! p^* f(1+\epsilon)$, where $p^* \!\defeq\! \min\{p_i:p_i>0\} \!>\! 0$. If $\mathcal{L}_\text{qc}(\sigma_+) \!-\! \mathcal{L}_\text{qc}(\sigma) \!<\! \delta$, then $\sum_i p_i f({q'_i}/q_i) \!<\! \delta.$ Thus, for every $i$ with $p_i>0$,
		\begin{equation}
			p_i f\!\left(\frac{q_i^{\prime}}{q_i}\right)
			\leqslant
			\sum\limits_j p_j f\!\left(\frac{q_j^{\prime}}{q_j}\right)
			<
			\delta
			\leqslant
			p_i f(1+\epsilon).
		\end{equation}
		Hence $f(q_i^{\prime}/q_i)<f(1+\epsilon)$. If $q_i^{\prime}/q_i<1$, then $q_i^{\prime}/q_i<1+\epsilon$ immediately. If $q_i^{\prime}/q_i\geqslant 1$, then monotonicity of $f$ on $[1,\infty)$ gives $q_i^{\prime}/q_i<1+\epsilon$. Therefore,
		\begin{equation}
			\Tr[\GX(\sigma)\, \sigma_+] = \sum\limits_i p_i\frac{q_i^{\prime}}{q_i} < 1 + \epsilon.
		\end{equation}
	\end{proof}
	
	\begin{lemma}\label{lem:norm}
		Let $\sigma$ be a $d\times d$ density matrix and let $A\in \mathbb{C}^{d\times d}$. Then $\|\sqrt{\sigma}A\sqrt{\sigma}\| \leqslant \|\sigma^{1/4}A\, \sigma^{1/4}\|$.
	\end{lemma}
	\begin{proof}
		In the eigenbasis of $\sigma$, we have $\sigma=\sum\nolimits_i \lambda_i \ketbra{i}{i}$. Then
		\begin{align}
			\begin{aligned}
				\|\sqrt{\sigma} A\sqrt{\sigma}\|^2 &= \sum_{i,j} \lambda_i \lambda_j |A_{ij}|^2 \leqslant \sum_{i,j} \sqrt{\lambda_i \lambda_j} |A_{ij}|^2 \\
				&= \|\sigma^{1/4} A\, \sigma^{1/4}\|^2.
			\end{aligned}
		\end{align}
	\end{proof}
	
	\begin{lemma}\label{lem:strong}
		For every $\epsilon>0$, there exists $\delta>0$ such that $\mathcal{L}_\text{qc}(\sigma_+)-\mathcal{L}_\text{qc}(\sigma)<\delta$ implies $\|\sigma_+-\sigma\|<\epsilon$.
	\end{lemma}
	\begin{proof}
		Let $\epsilon>0$. By Lemmas \ref{lem:control} and \ref{lem:norm}, there exists $\delta>0$ such that
		\begin{align}
			\|\sigma_+-\sigma\|^2 &= \|\sqrt{\sigma}(\GX(\sigma)-\mathds{1})\sqrt{\sigma}\|^2 \nonumber \\
			&\leqslant \|\sigma^{1/4}(\GX(\sigma)-\mathds{1})\, \sigma^{1/4}\|^2 \nonumber \\
			&= \Tr[(\GX(\sigma)-\mathds{1}) \sqrt{\sigma}(\GX(\sigma)-\mathds{1})\sqrt{\sigma}] \nonumber \\
			&= \Tr[(\GX(\sigma)-\mathds{1})(\sigma_+-\sigma)] \nonumber \\
			&= \Tr[\GX(\sigma)(\sigma_+-\sigma)] \\
			&= \Tr[\GX(\sigma)\, \sigma_+] - 1 < \epsilon^2 \nonumber
		\end{align}
		whenever
		$\mathcal{L}_\text{qc}(\sigma_+)-\mathcal{L}_\text{qc}(\sigma)<\delta$. For the penultimate equality, we used the fact that $\sigma_+-\sigma$ is traceless.
	\end{proof}
	It remains to prove monotonicity and the equality condition; this is the content of Theorem~\ref{thm:equality} below, whose proof uses Lemma~\ref{lem:strong}.
	
	\begin{theorem}\label{thm:equality}
		$\mathcal{L}_\text{qc}(\sigma_+) \geqslant \mathcal{L}_\text{qc}(\sigma)$ with equality iff $\sigma_+=\sigma$.
	\end{theorem}
	\begin{proof}
		Suppose $\mathcal{L}_\text{qc}(\sigma_+)-\mathcal{L}_\text{qc}(\sigma) \leqslant 0$. Then $\mathcal{L}_\text{qc}(\sigma_+) - \mathcal{L}_\text{qc}(\sigma) < \delta$ for any $\delta>0$. By Lemma~\ref{lem:strong}, $\|\sigma_+-\sigma\|=0$ and $\sigma_+=\sigma$. Thus, it is not possible that $\mathcal{L}_\text{qc}(\sigma_+) < \mathcal{L}_\text{qc}(\sigma)$ and $\mathcal{L}_\text{qc}(\sigma_+) = \mathcal{L}_\text{qc}(\sigma)$ implies $\sigma_+=\sigma$.
	\end{proof}
	Together with the explicit formula for the Petz update established above, this proves Theorem~\ref{thm:Petzqc}.
	
	\begin{proof}[Proof of Theorem~\ref{thm:Petz.converg}]
		By Theorem~\ref{thm:equality}, $\mathcal{L}_{\text{qc}}(\sigma_k)$ is nondecreasing. By Lemma~\ref{lem:qc.unique.MLE}, $\mathcal{L}_{\text{qc}}$ has a maximizer, so $\mathcal{L}_{\text{qc}}(\sigma_k)$ is bounded above. Hence $\lim_{k\to\infty}\mathcal{L}_{\text{qc}}(\sigma_k) \eqdef \ell_\infty$ exists. Let $\mathcal{C}\defeq \overline{\{\sigma_k\}}$. Since
		$\mathcal{L}_{\text{qc}}(\sigma_k)\geqslant \mathcal{L}_{\text{qc}}(\sigma_0)>-\infty$, the probabilities $\Tr[\Pi_i\sigma_k]$ are bounded away from zero $\forall\, i$. Thus every point of $\mathcal{C}$ lies in the domain where the Petz update and $\mathcal{L}_{\text{qc}}$ are continuous. Define
		\begin{equation}
			F(\sigma) \defeq \|\sigma_+-\sigma\| + |\mathcal{L}_{\text{qc}}(\sigma)-\ell_\infty|
		\end{equation}
		on this domain, and set $\mathcal{A}\defeq\mathcal{C}\cap F^{-1}(0)$. By Lemma~\ref{lem:strong} and the definition of $\ell_\infty$, $F(\sigma_k)\to0$. Thus, by compactness of $\mathcal{C}$ and continuity of $F$, there exists $\sigma\in\mathcal{C}$ such that $F(\sigma)=0$, and $\mathcal{A}$ is nonempty.
		
		\textit{Case 1}. Suppose there exists an invertible state $\sigma \!\in\! \mathcal{A}$. Since $F(\sigma) \!=\! 0$, we have $\sigma_+ \!=\! \sigma$. As $\sigma$ is invertible, this fixed-point equation implies $\GX(\sigma) \!=\! \mathds{1}$. Hence
		\begin{equation}
			D\mathcal{L}_{\text{qc}}(\sigma)[\Delta]=0
		\end{equation}
		for every traceless Hermitian $\Delta$, and by concavity $\sigma$ maximizes $\mathcal{L}_{\text{qc}}$. By Lemma~\ref{lem:qc.unique.MLE}, $\sigma \!=\! \hat{\sigma}_{\text{MLE}}$. Since $\mathcal{L}_{\text{qc}}(\sigma) \!=\! \ell_\infty$, every $\omega \!\in\! \mathcal{A}$ also satisfies
		\begin{equation}
			\mathcal{L}_{\text{qc}}(\omega) = \ell_\infty = \mathcal{L}_{\text{qc}}(\hat{\sigma}_{\text{MLE}}),
		\end{equation}
		and is therefore a maximizer. Again by Lemma~\ref{lem:qc.unique.MLE}, $\omega=\hat{\sigma}_{\text{MLE}}$, so $\mathcal{A}=\{\hat{\sigma}_{\text{MLE}}\}$.
		
		It follows that $\sigma_k\to\hat{\sigma}_{\text{MLE}}$. If not, then there exists an open set $U$ containing $\hat{\sigma}_{\text{MLE}}$ and a subsequence such that $\sigma_{k_i} \!\notin\! U \;\forall\, i$. Then the closure $\mathcal{B}$ of $\{\sigma_{k_i}\}$ is a compact subset of $\mathcal{C}$ and does not contain $\hat{\sigma}_{\text{MLE}}$. Thus $\mathcal{B}\cap\mathcal{A}=\emptyset$, and by compactness of $\mathcal{B}$ and continuity of $F$ on $\mathcal{C}$,
		\begin{equation}
			F(\sigma_{k_i})\geqslant \min_{\sigma\in\mathcal{B}}F(\sigma)>0,
		\end{equation}
		contradicting $F(\sigma_k)\to0$.
		
		\textit{Case 2}. Suppose $\det(\sigma) \!=\! 0 \;\forall\,\sigma \!\in\! \mathcal{A}$. Then $\det(\sigma_k) \!\to\! 0$. If not, then there exists $\epsilon \!>\! 0$ and a subsequence $\sigma_{k_i}$ such that $\det(\sigma_{k_i}) \!>\! \epsilon \;\forall\, i$. Then the closure $\mathcal{B}$ of $\{\sigma_{k_i}\}$ is a compact subset of $\mathcal{C}$ and, by continuity of the determinant, does not intersect $\mathcal{A}$. Thus, by compactness of $\mathcal{B}$ and continuity of $F$ on $\mathcal{C}$,
		\begin{equation}
			F(\sigma_{k_i})\geqslant \min_{\sigma\in\mathcal{B}}F(\sigma)>0,
		\end{equation}
		contradicting $F(\sigma_k)\to0$.
		
		This proves the dichotomy. It remains to prove the level set assertion. Suppose that the upper level set
		\begin{equation}
			K = \{\sigma:\mathcal{L}_{\text{qc}}(\sigma) \geqslant \mathcal{L}_{\text{qc}}(\sigma_0)\}
		\end{equation}
		contains only invertible states. By monotonicity, $\sigma_k \!\in\! K \;\forall\, k$. Since $K$ is compact and $K\cap\det^{-1}(0)=\emptyset$,
		\begin{equation}
			\min_{\sigma\in K}\det(\sigma)>0.
		\end{equation}
		Thus the alternative $\det(\sigma_k)\to0$ is impossible. By the dichotomy just proved, $\sigma_k\to\hat{\sigma}_{\text{MLE}}$.
	\end{proof}
	
	Note that this theorem extends to any number of measurements by the block channel construction outlined in the following subsection.
	
	\subsection{Weighted channel data as a single block channel} \label{sm:sec:block.channel}
	A weighted likelihood for several channels can be written as a single-channel likelihood on a direct-sum output space. Define
	\begin{equation}
		\mathcal{E}_{\oplus}(\sigma) = \bigoplus_j w_j\,\mathcal{E}_j(\sigma),
		\quad
		X_{\oplus} = \bigoplus_j w_j\,X_j,
		\quad
		\sum_j w_j = 1,
	\end{equation}
	and $w_j\geqslant 0$.
	Then, up to the constant $\sum_j w_j \Tr[X_j]\ln w_j$,
	\begin{equation}
		\Tr[X_{\oplus}\ln\mathcal{E}_{\oplus}(\sigma)]
		=
		\sum_j w_j \Tr[X_j\ln\mathcal{E}_j(\sigma)].
	\end{equation}
	Moreover, since $\Omega^{-1}_{w_jY_j}(w_jX_j)=\Omega^{-1}_{Y_j}(X_j)$, the KM update for $\mathcal{E}_{\oplus}$ is exactly
	\begin{equation}
		\mathcal{R}^{\text{KM}}_{\sigma,\mathcal{E}_{\oplus}}(X_{\oplus})
		=
		\sum_j w_j\,\mathcal{R}^{\text{KM}}_{\sigma,\mathcal{E}_j}(X_j).
	\end{equation}
	
	\paragraph*{Example: quantum-to-classical channel.} \mbox{} \\
	For $\mathcal{E}_j(\sigma) = \sum_i \Tr \big[\Pi_i^{(j)}\sigma\big] \ketbra{j,i}{j,i}$, the block channel $\mathcal{E}_{\oplus}$ is the quantum-to-classical channel associated with the single POVM $\{w_j\Pi_i^{(j)}\}_{i,j}$ and the block data
	\begin{equation}
		X_{\oplus} = \sum_{j,i} w_j\hat{p}_i^{(j)} \ketbra{j,i}{j,i}.
	\end{equation}
	Thus the weighted Petz update is the ordinary Petz update for this unified POVM, and the weighted KM identity above reduces to the same construction in the commuting case.
	
	\subsection{Classical and quantum Bayes' theorem and monotonicity violation by the Petz map} 
	\label{app:subsec:quantum.Bayes.theorem}
	Let $X$ and $Y$ be random variables taking values from finite sets $\mathcal{X}$ and $\mathcal{Y}$ respectively. We view the random variable $X$ as representing the state of a classical system, while $Y$ is the result of an observation on the system. The observation can be modeled as a ``forward'' process $\mathcal{E}$ corresponding to the stochastic matrix with entries given by the conditional probabilities $p(y|x)$ for $x \!\in\! \mathcal{X}$ and $y \!\in\! \mathcal{Y}$. Given a prior distribution $\sigma(x)$ representing an agent's belief about $X$, we apply the forwards process to obtain the distribution
	\begin{equation}
		\mathcal{E}(\sigma)(y) = \sum_{x\in\mathcal{X}} p(y|x)\, \sigma(x).
	\end{equation}
	If an observation $y$ is made, Bayes' rule prescribes the following update to the prior: 
	\begin{equation}\label{eq:bayes}
		\sigma_y(x) = \frac{p(y|x)\, \sigma(x)}{\sum_{x\in\mathcal{X}}\, p(y|x)\, \sigma(x)}.
	\end{equation}
	Jeffrey's rule extends Bayes' to situations with ``soft evidence'', where the observation is a distribution $\tau(y)$. In this case, the updated state is
	\begin{equation}
		\sigma_+(x) = \sum_{y\in\mathcal{Y}} \sigma_y(x)\, \tau(y).
		\label{eq:jeffrey}
	\end{equation}
	The expression above has the form $\sigma_+=\mathcal{R}_{\sigma,\mathcal{E}}(\tau)$, where $\mathcal{R}_{\sigma,\mathcal{E}}$ is the Bayesian reverse process determined by the prior $\sigma$ and the forward channel $\mathcal{E}$. Bayes' rule \eqref{eq:bayes} is recovered when $\tau$ is a delta distribution. Jeffrey's rule follows from Jeffrey's probability kinematics \cite{Jeffrey:Book:1990}, is equivalent to Pearl's virtual-evidence update \cite{Pearl:Book:1988,Chan.Darwiche:ArtifIntell:2005}, and can also be derived from a minimum-change principle \cite{Bai.Buscemi.Scarani:PhysRevLett:2025}.
	
	Viewing quantum mechanics as the noncommutative generalization of classical probability theory, one naturally asks for a quantum Bayes rule: a CPTP map $\mathcal{R}$ extending \eqref{eq:jeffrey} to quantum states and channels. Several inequivalent proposals exist \cite{Parzygnat.Russo:LinearAlgebrAppl:2022,Surace.Scandi:Quantum:2023,Parzygnat.Buscemi:Quantum:2023,Bai.Buscemi.Scarani:PhysRevLett:2025}, reflecting the fact that no single quantum extension has the same universal status as classical Bayes' rule. In the relative-entropy setting, the canonical object is the \textit{Petz recovery map} \eqref{def:Petz.update}, which characterizes equality in the data-processing inequality \cite{Petz:CommunMathPhys:1986,Petz:QuartJMathOxfordSer:1988,Petz:RevMathPhys:2003}
	\begin{equation}
		D(\mathcal{E}(\rho)\|\mathcal{E}(\sigma)) \leqslant D(\rho\|\sigma).
		\label{eq:data.processing.inequality}
	\end{equation}
	For faithful $\sigma$ and compatible supports, equality holds precisely when $\rho$ is recovered from its channel output by the Petz map, $\mathcal{R}_{\sigma,\mathcal{E}}(\mathcal{E}(\rho))=\rho$.
	
	Jacobs \cite{Jacobs:ElectronProcTheorComputSci.2021} showed that Jeffrey's update moves the predicted observation toward the soft evidence in a precise relative-entropy sense: for the Kullback--Leibler divergence \cite{Kullback.Leibler:AnnMathStat:1951}
	\begin{equation}
		D_{\text{KL}}(p\|q) \defeq \sum_z p(z)[\ln p(z)-\ln q(z)],
	\end{equation}
	one has
	\begin{equation}
		D_{\text{KL}}(\tau\|\mathcal{E}(\sigma_+)) 
		\leqslant
		D_{\text{KL}}(\tau\|\mathcal{E}(\sigma)),
	\end{equation}
	where $\tau$, $\mathcal{E}(\sigma_+)$, and $\mathcal{E}(\sigma)$ are classical distributions.
	
	Theorem~\ref{thm:Petzqc} in the main text generalizes the Petz--Jeffrey monotonicity statement for quantum-to-classical channels~\eqref{eq:qc.channel}, whose outputs are mutually commuting, even though their quantum inputs need not be. This Theorem can be rewritten as,
	\begin{equation}			D\!\left(X\,\middle\|\,\mathcal{E}_{\text{qc}}(\sigma_+)\right)
		\leqslant
		D\!\left(X\,\middle\|\,\mathcal{E}_{\text{qc}}(\sigma)\right),
	\end{equation}
	with equality iff $\sigma_+=\sigma$, and where $\sigma_+
	=
	\mathcal{R}_{\sigma,\mathcal{E}_{\text{qc}}}(X)$.
	
	However, as mentioned in the main text, the Petz recovery map does not satisfy such a monotonicity theorem in general. The obstruction is the noncommutative output differential: when $[\mathcal{E}(\sigma),X]\neq0$, one generally has
	\begin{equation}
		\Omega_{\mathcal{E}(\sigma)}^{-1}\!(X) \neq \mathcal{E}(\sigma)^{-1/2}X\,\mathcal{E}(\sigma)^{-1/2}.
	\end{equation}
	A concrete example which exhibits violation of this inequality is
	\begin{align}
		\sigma &=
		\begin{pmatrix}
			\frac{2}{3} & 0 \\
			0 & \frac{1}{3}
		\end{pmatrix}\!,\ 
		X =
		\begin{pmatrix}
			0.99 & 0 \\
			0 & 0.01
		\end{pmatrix}\!,\ 
		\mathcal{E}(\sigma) = \sum_{j=1}^{2} K_j \sigma K_j^\dagger,\nonumber \\
		K_1 &\!=\!
		\begin{pmatrix}
			-0.251899 \!+\! 0.447484\,i &
			\;\;0.284823 \!-\! 0.031868\,i
			\\
			0.308714 \!-\! 0.635360\,i &
			-0.388139 \!+\! 0.067281\,i
		\end{pmatrix}\!,\nonumber \\
		K_2 &\!=\!
		\begin{pmatrix}
			-0.215046 \!+\! 0.189193\,i &
			-0.500736 \!-\! 0.114805\,i
			\\
			0.279643 \!-\! 0.277632\,i &
			0.696647 \!+\! 0.115962\,i
		\end{pmatrix}\!.
	\end{align}
	For the Petz recovery map $\mathcal{R}_{\sigma,\mathcal{E}}$, we obtain 
	\begin{equation}
		D\!\left(X\,\middle\|\,\mathcal{E}(\sigma)\right)-D\!\left(X\,\middle\|\,\mathcal{E}(\mathcal{R}_{\sigma,\mathcal{E}}(X))\right)=-2.022036,
	\end{equation}
	while for the KM update $\mathcal{R}^{\text{KM}}_{\sigma,\mathcal{E}}$, we have
	\begin{equation}
		D\!\left(X\,\middle\|\,\mathcal{E}(\sigma)\right)-D\!\left(X\,\middle\|\,\mathcal{E}(\mathcal{R}^{\text{KM}}_{\sigma,\mathcal{E}}(X))\right)=0.069089.
	\end{equation}  
	Notice that in this case $X$ and $Y=\mathcal{E}(\sigma)$ are highly non-commuting; the Frobenius norm is equal to $\norm{[X,Y]}_F=0.659206$, which is  close to the maximum given by $1/\sqrt{2}$. This example was found numerically using random instances further refined by simulated annealing.
	
	This is consistent with Theorem~1 of Ref.~\cite{Bai.Buscemi.Scarani:PhysRevLett:2025}, where the Petz map is recovered as the solution of the minimum-change problem in the output-commuting case $[\mathcal{E}(\sigma),X]=0$. The corresponding noncommutative likelihood-gradient correction is the inverse KM action derived in the End Matter.
	
	\subsection{Expressions for the quantum-to-classical channel}
	The quantum-to-classical channel is defined as
	\begin{equation}
		\E_{\text{qc}}(\sigma) = \sum_i \Tr[\Pi_i \sigma] \ketbra{i}{i}.
		\label{app:eq:qc.channel}
	\end{equation}
	Using $\Pi_i=K_i^\dag K_i$, we can rewrite it in terms of its Kraus decomposition as
	\begin{equation}
		\E_{\text{qc}}(\sigma) = \sum_{i,j} \ketbra{i}{j} K_i \sigma K_i^\dag\ketbra{j}{i},
	\end{equation}
	with adjoint
	\begin{equation}
		\E_{\text{qc}}^\dag(Z) = \sum_{i,j} K_i^\dag\ketbra{j}{i} Z \ketbra{i}{j} K_i.
	\end{equation}
	
	The inverse KM superoperator is 
	\begin{equation}
		\Omega_{\E_{\text{qc}}(\sigma)}^{-1}\!(X) = \int_{\mathbb{R}} \beta_0(t)\, \E_{\text{qc}}(\sigma)^{-(1+it)/2}X\,\E_{\text{qc}}(\sigma)^{-(1-it)/2}\,dt.
	\end{equation}
	If $X$ is diagonal in the $\{\ket{i}\}$ basis, e.g., because it has been produced by the quantum-to-classical channel, then $X=\sum_i \hat{p}_i\ketbra{i}{i}$, $\E_{\text{qc}}(\sigma)$ and $X$ commute, and
	\begin{equation}
		\Omega_{\E_{\text{qc}}(\sigma)}^{-1}\!(X) = \E_{\text{qc}}(\sigma)^{-1/2}X\,\E_{\text{qc}}(\sigma)^{-1/2}=\sum_i \frac{\hat{p}_i}{ \Tr[\Pi_i \sigma]}\ketbra{i}{i}.
	\end{equation}
	This means that the Hilbert-Schmidt gradient is
	\begin{align}
		\GX(\sigma) &= \mathcal{E}^{\dagger}\!\big[\Omega_{\E_{\text{qc}}(\sigma)}^{-1}\!(X)\big]
		= \sum_{i,j} K_i^\dag\ket{j} \frac{\hat{p}_i}{ \Tr[\Pi_i \sigma]} \bra{j} K_i \nonumber \\
		&= \sum_{i} \frac{\hat{p}_i}{ \Tr[\Pi_i \sigma]} \Pi_i,
	\end{align}
	and the gradient in the inverse square-root metric \eqref{def:inverse.square.root.metric} is
	\begin{equation}
		G_{\!X}^\sigma(\sigma)= \sqrt{\sigma}\, \GX(\sigma) \sqrt{\sigma}=\sum_{i} \frac{\hat{p}_i}{ \Tr[\Pi_i \sigma]} \sqrt{\sigma}\, \Pi_i \sqrt{\sigma}.
	\end{equation}
	This leads to the Riemannian gradient,
	\begin{equation}
		\Delta = G_{\!X}^\sigma(\sigma)-\sigma
		= \sqrt{\sigma}\Big(\sum_{i} \frac{\hat{p}_i}{ \Tr[\Pi_i \sigma]} \Pi_i - I \Big) \sqrt{\sigma}.
	\end{equation}
	
	If $X=\E_{\text{qc}}(\rho)$ has been computed directly as an output of the channel for a state $\rho$, then $\hat{p}_i\equiv p_i=\Tr[\Pi_i\rho]$, and the KM update applied to $\E_{\text{qc}}(\rho)$ equals the Petz recovery map,
	\begin{equation}
		\mathcal{R}^{\text{KM}}_{\sigma,\E_{\text{qc}}}(\E_{\text{qc}}(\rho)) 
		= \mathcal{R}_{\sigma,\E_{\text{qc}}}(\E_{\text{qc}}(\rho)) 
		= \sum_{i} \frac{p_i}{\Tr[\Pi_i \sigma]} \sqrt{\sigma}\, \Pi_i \sqrt{\sigma}.
	\end{equation}
	
	Inserting the quantum-to-classical channel and $X$ into the generalized log-likelihood, we obtain
	\begin{equation}
		\begin{split}
			\mathcal{L}_X(\sigma)
			&= \Tr\!\big[X \ln \E_{\text{qc}}(\sigma)\big] \\
			&= \Tr\!\bigg[\sum_i \hat{p}_i\ketbra{i}{i} \ln \Big(\sum_i \Tr[\Pi_i \sigma] \ketbra{i}{i}\Big)\bigg] \\
			&=\sum_i \hat{p}_i \ln \Tr[\Pi_i \sigma].
		\end{split}
	\end{equation}
	This is the standard log-likelihood.
	
	For a line search update
	\begin{equation}
		\sigma_\eta = \sigma + \eta\, \Delta,
		\qquad 
		f(\eta) \defeq \mathcal{L}_X(\sigma_\eta),
	\end{equation}
	the quantum-to-classical channel gives
	\begin{align}
		\begin{aligned}
			f'(0) &= \textsl{g}_\sigma(\Delta,\Delta) = \|\Delta\|_{\textsl{g}_\sigma}^2\\ &=\Tr[\left(\bigg(\sum_{i} \frac{\hat{p}_i}{ \Tr[\Pi_i \sigma]} \Pi_i - I \bigg)\sqrt{\sigma}\right)^2] .
		\end{aligned}		
	\end{align}
	This means an approximate increase in the log-likelihood when performing a line search update,
	\begin{equation}    
		\mathcal{L}_X(\sigma_\eta)\approx \mathcal{L}_X(\sigma) + \eta\, f'(0).
	\end{equation}
	
	\subsection{Thermodynamically-motivated KM update example}
	
	\begin{figure}[htbp!]
		\centering
		\includegraphics[width=.98\linewidth]{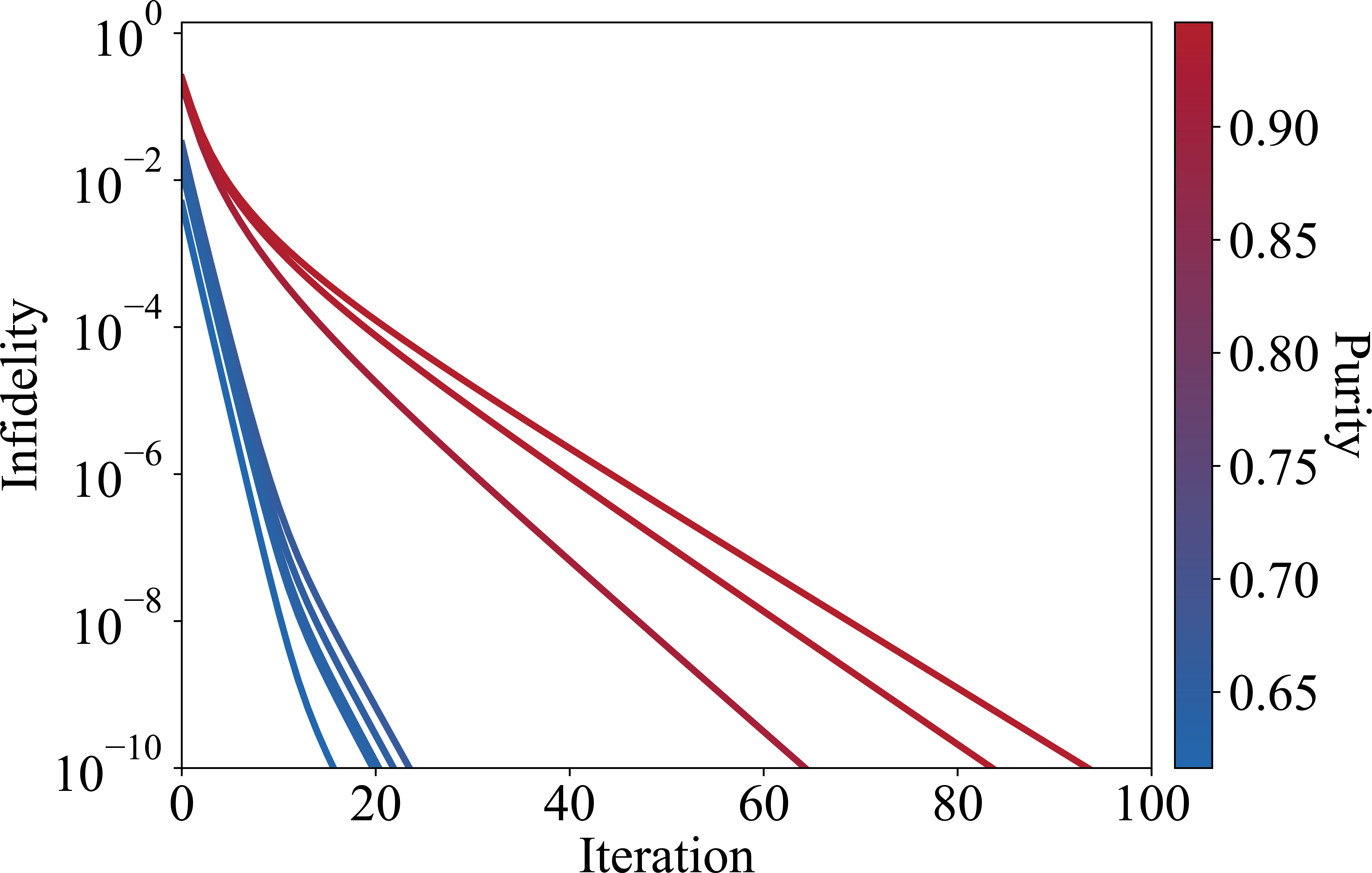}
		\begingroup
		\setlength{\abovecaptionskip}{2mm}
		\setlength{\belowcaptionskip}{0mm}
		\caption{KM reconstruction for a qubit undergoing thermal Lindblad evolution for $t=0.5\,\gamma_{\downarrow}^{-1}$, with $\beta\omega=1$ and $\gamma_{\uparrow}/\gamma_{\downarrow}=e^{-\beta\omega}$. Only $\langle\sigma_x\rangle$ and $\langle\sigma_y\rangle$ of the evolved state are assumed to be measured, while the missing diagonal is inferred from the mean energy fixed by the bath temperature, $\Tr[HX]=\Tr[H\tau_\beta]$. The KM iteration starts from the thermal state $\tau_\beta$ and reconstructs the initial state. Curves show the infidelity with the analytic inverse for ten random initial states with the same mean energy and different coherences. The color indicates their purity.}
		\label{fig:thermodynamic}
		\endgroup
	\end{figure}	
	
	We consider a qubit with Hamiltonian
	\begin{equation}
		H=\omega |1\rangle\langle 1|
	\end{equation}
	coupled to a thermal bath at inverse temperature $\beta$. Its dynamics is described by the master equation~\cite{Breuer.Petruccione:book:2002},
	\begin{equation}
		\dot{\rho}
		=
		-i[H,\rho]
		+
		\gamma_{\downarrow}\mathcal{D}[\sigma_-]\rho
		+
		\gamma_{\uparrow}\mathcal{D}[\sigma_+]\rho,
	\end{equation}
	with detailed balance
	\begin{equation}
		\frac{\gamma_{\uparrow}}{\gamma_{\downarrow}}
		=
		e^{-\beta\omega}.
	\end{equation}
	The corresponding stationary state is the Gibbs state
	\begin{equation}
		\tau_\beta
		=
		\frac{e^{-\beta H}}{\Tr[e^{-\beta H}]}
		=
		\begin{pmatrix}
			1-p_\beta & 0\\
			0 & p_\beta
		\end{pmatrix},
		\qquad
		p_\beta=\frac{1}{1+e^{\beta\omega}}.
	\end{equation}
	
	We generate initial states with the same mean energy as the bath,
	\begin{equation}
		\Tr[H\rho]=\Tr[H\tau_\beta]=\omega p_\beta,
	\end{equation}
	but with different coherences, and evolve them for a fixed time $t$ under the channel $\mathcal{E}_t=e^{t\mathcal{L}}$. At the output we assume that only $\langle\sigma_x\rangle$ and $\langle\sigma_y\rangle$ are measured, determining the off-diagonal matrix element of the state $X$. The missing diagonal is completed using the known bath temperature through the constraint
	\begin{equation}
		\Tr[HX]=\omega p_\beta.
	\end{equation}
	Thus
	\begin{equation}
		X=
		\begin{pmatrix}
			1-p_\beta & (x-iy)/2\\
			(x+iy)/2 & p_\beta
		\end{pmatrix},
	\end{equation}
	where $x=\langle\sigma_x\rangle$ and $y=\langle\sigma_y\rangle$.
	
	We then infer the initial state by iterating the KM update for the channel $\mathcal{E}_t$, see Fig.~\ref{fig:thermodynamic}, starting from the thermal prior
	\begin{equation}
		\sigma_0=\tau_\beta.
	\end{equation}
	Since $X$ generally contains coherence while $\mathcal{E}_t(\tau_\beta)=\tau_\beta$ is diagonal, $[X,\mathcal{E}_t(\sigma_0)]\neq0$, so the iteration is genuinely non-commutative. The reconstructed state is compared with the analytic inverse of the thermal channel using the fidelity.	

\end{document}